\documentclass[10pt]{article}
\usepackage[centertags]{amsmath}
\usepackage{amsfonts}
\usepackage{amssymb}
\usepackage{amsthm}
\input{amssym.def}
\usepackage{tikz}
\usepackage{graphicx}
\usepackage{faktor}

\newtheorem{lem}{Lemma}
\newtheorem{prop}[lem]{Proposition}
\newtheorem{theorem}[lem]{Theorem}
\newtheorem{coro}[lem]{Corollary}

\theoremstyle{definition}

\newtheorem{definition}[lem]{Definition}

\newtheorem{lemma}[lem]{Lemma}
\newtheorem{remark}[lem]{Remark}

\begin{document}
\date{}

\title{Function-Based Minimal Linear Codes over Galois Rings $\mathrm{GR}(p^{n}, \ell)$: Minimality Criteria and Infinite Constructions}
\author{Biplab Chatterjee$^1$, Sihem Mesnager$^2$, Ratnesh Kumar Mishra$^3$, \\
Makhan Maji$^4$, and Kalyan Hansda$^5$  \\
\footnotesize{Department of Mathematics, NIT Jamshedpur, Jamshedpur, 831014, Jharkhand, India$^{1,3}$}\\
\footnotesize{Department of Mathematics, University of Paris VIII, 93526 Saint-Denis, LAGA, UMR 7539, CNRS,}\\
\footnotesize{93430 Villetaneuse, and Telecom Paris, 91120 Palaiseau, Paris, France$^2$}\\
\footnotesize{Indian Institute of Technology, Madras, Chennai, Tamil Nadu 600036, India$^4$} \\
\footnotesize{Department of Mathematics, Visva-Bharati, Santiniketan, Bolpur-731235, West Bengal, India$^{5}$}\\
\footnotesize{2022rsma001@nitjsr.ac.in$^{1}$} , \
\footnotesize{smesnager@univ-paris8.fr$^{2}$}, \
\footnotesize{ratnesh.math@nitjsr.ac.in$^{3}$}, \
\footnotesize{makhan2maths@gmail.com$^{4}$} and \
\footnotesize{kalyanh4@gmail.com$^{5}$}}

\maketitle

\begin{abstract}
In this paper, we extend a necessary and sufficient condition for a linear code over a Galois ring to be minimal and establish new bounds on the length of an $m$-dimensional minimal linear code.

Building upon this structural characterization, we further generalize the function-based minimality criteria introduced by Wu \emph{et al.} (Cryptogr. Commun. 14, 875-895, 2022) from the finite field setting to the framework of Galois rings. The transition from fields to rings introduces substantial algebraic challenges due to the presence of zero divisors and the richer module structure of $\mathrm{GR}(p^{n},\ell)$.

By exploiting Frobenius duality and the chain structure of Galois rings, we derive refined necessary and sufficient conditions ensuring that linear codes arising from functions over $\mathrm{GR}(p^{n},\ell)$ are minimal. As an application of these criteria, we construct several infinite families of minimal linear codes over Galois rings, thereby significantly generalizing the constructions of Wu \emph{et al.} to the ring setting.
Our results provide a unified framework that connects minimality theory, module duality over Frobenius rings, and function-based code constructions.
\end{abstract}

\section{Introduction}
Minimal linear codes constitute a fundamental and extensively studied class of linear codes, owing to their intrinsic connection to secret sharing schemes and their broad range of cryptographic applications. In particular, the minimality property of codewords is intimately related to the access structure of a secret sharing scheme: minimal codewords correspond precisely to minimal qualified subsets. Consequently, minimal codes play a central role in secure information distribution, multiparty computation (MPC), and verifiable secret sharing (VSS), where fine structural control over access structures is essential.

In recent years, increasing attention has been devoted to minimal codes defined over finite rings. In contrast to the field setting, finite rings, especially finite chain rings and Frobenius rings, possess a richer algebraic structure arising from the presence of nontrivial ideals and zero divisors. This additional structural complexity provides greater flexibility in code design and may lead to enhanced performance in cryptographic protocols and distributed storage systems~\cite{Cao2019,Meletiou2025}. Furthermore, ring-based constructions arise naturally in post-quantum cryptography, particularly in lattice-based schemes, and they exhibit promising potential in decentralized environments such as blockchain infrastructures and secure electronic voting systems~\cite{Munera-Merayo}.

Over finite fields, several infinite families of minimal linear codes have been constructed using diverse algebraic, combinatorial, and geometric techniques. Significant contributions in this direction were obtained by Ding \emph{et al.}~\cite{DingHengZhou}, Heng \emph{et al.}~\cite{HengDingZhou}, Bartoli and Bonini~\cite{BartoliBonini}, Mesnager \emph{et al.}~\cite{MesnagerQiRuTang}, and Bonini and Borello~\cite{BoniniBorello}. A particularly powerful and flexible construction paradigm was later introduced by Wu \emph{et al.}~\cite{WuLuCao}, in which minimal linear codes are derived from generator matrices defined by suitable functions. This function-based framework unifies numerous earlier constructions and considerably enlarges the spectrum of admissible defining functions, thereby offering a systematic approach to the production of infinite families of minimal codes.

Minimal linear codes over finite commutative rings have also been
investigated in recent years. In particular, one-dimensional
minimal linear codes over finite rings were characterized
in~\cite{MajiMesnager}. Furthermore, in~\cite{ChatterjeeMesnager}, necessary
and sufficient conditions were established for a linear code over
$\mathbb{Z}_n$ to be minimal, providing a complete minimality
criterion in the ring setting.

Motivated by the aforementioned developments, the present work is devoted to the systematic study of minimal linear codes over the Galois ring $\mathrm{GR}(p^{n},\ell)$. Our first objective is to establish a complete characterization of minimality for linear codes defined over $\mathrm{GR}(p^{n},\ell)$, extending to this setting the approach developed in~\cite{ChatterjeeMesnager}.
A crucial ingredient in our analysis is the Frobenius structure of $\mathrm{GR}(p^{n},\ell)$. Since Galois rings are finite Frobenius rings, powerful duality properties become available. In particular, for any linear code $C$ over a Frobenius ring, one has the double-orthogonality property $(C^{\perp})^{\perp} = C,$ which follows from the fundamental identity $|C|\,|C^{\perp}| = |R^{m}|$ established in~\cite{Wood}. This duality principle plays a central structural role in deriving necessary and sufficient conditions for minimality.  As a consequence of this characterization, we obtain explicit lower bounds on the length of $k$-dimensional minimal linear codes over $\mathrm{GR}(p^{n},\ell)$. We further analyze the structure of one-dimensional minimal codes in this framework, thereby extending and refining the results of~\cite{ChatterjeeMesnager} to the Galois ring setting.

The second main objective of this paper is constructive in nature. We develop explicit constructions of minimal linear codes over $\mathrm{GR}(p^{n},\ell)$ within the function-based generator matrix framework introduced by Wu \emph{et al.}~\cite{WuLuCao}. More precisely, we establish necessary and sufficient minimality conditions for linear codes arising from $p^{n\ell}$-ary functions defined on $\mathrm{GR}(p^{n},\ell)^{m}$. A major technical difficulty in the ring setting stems from the presence of zero divisors, which gives rise to phenomena absent in the finite field case. In particular, the classification of root words becomes substantially richer and requires a refined module-theoretic analysis. We show that the minimality of an $(m+1)$-dimensional linear code can be effectively characterized through the structural behavior of root words and their associated orthogonal modules. This perspective not only clarifies the role of zero divisors but also significantly streamlines the analysis of minimality.

Throughout this paper, we denote $q = p^{\ell}$.

In order to construct an $(m+1)$-dimensional minimal linear code over the Galois ring $\mathrm{GR}(p^{n},\ell)$ via the function-based framework, a fundamental structural problem arises: namely, how should the defining $q^{n}$-ary function be chosen to guarantee minimality?

By construction (see Section~\ref{Section 4}), the first row of the generator matrix $G_{\Lambda_f}$ of the code $C_f$ consists precisely of the values of the $q^{n}$-ary function $f$ over $\mathrm{GR}(p^{n},\ell)^{m}$. If $C_f$ is an $(m+1)$-dimensional minimal linear code, then the one-dimensional code generated by this first row must itself be minimal. Consequently, Proposition~\ref{Proposition 27} implies that this row must contain at least one generator of each non-zero proper ideal of the Galois ring $\mathrm{GR}(p^{n},\ell)$.

Recall that the non-zero proper ideals of $\mathrm{GR}(p^{n},\ell)$ are of the form $\langle p^{r} \rangle$, where $1 \le r \le n-1$ (see Section~\ref{Section 3}), and that each such ideal is generated by elements of the form $p^{r}u_j$, with $u_j$ a unit in $\mathrm{GR}(p^{n},\ell)$. Therefore, the set of function values appearing in the first row of $G_{\Lambda_f}$ must include representatives of the form $p^{r}u_j$ for every admissible exponent $r$. This observation leads to the central constructive question: for which vectors
$v \in \mathrm{GR}(p^{n},\ell)^{m}$ does the function satisfy $f(v) = p^{r}u_j \, ?$
In other words, one must determine and classify all vectors $v$ whose functional image lies in each prescribed non-zero proper ideal of the Galois ring.
By resolving this structural problem, we are able to construct infinite families of minimal linear codes over $\mathrm{GR}(p^{n},\ell)$. Moreover, this analysis enables us to identify suitable and sufficiently large classes of defining functions on $\mathrm{GR}(p^{n},\ell)^{m}$, thereby yielding several infinite families of minimal linear codes over Galois rings.

The remainder of this paper is organized as follows. In Section~\ref{preliminaries}, we recall the fundamental definitions and preliminary results that will be used throughout the paper.
Section~\ref{Section 3} is devoted to the algebraic structure of the Galois ring $\mathrm{GR}(p^{n},\ell)$, where we develop the module-theoretic framework and structural properties that form the foundation of our analysis.
In Section~\ref{Section 4}, we establish necessary and sufficient conditions for the minimality of linear codes over $\mathrm{GR}(p^{n},\ell)$ and derive explicit lower bounds on the length of $k$-dimensional minimal linear codes.
Section~\ref{section 5} presents the function-based construction method and introduces several infinite families of minimal linear codes over $\mathrm{GR}(p^{n},\ell)$, thereby illustrating the effectiveness and applicability of the theoretical results obtained in the preceding sections.
Finally, Section~\ref{sec:conclusion} concludes the paper with a summary of our main contributions and a discussion of possible directions for future research.

\section{Preliminaries and Background on Linear Codes over Rings} \label{preliminaries}

\subsection{Linear Codes over Finite Fields and Finite Rings}

In this section, we introduce the fundamental notions and terminology related to linear codes that will be used throughout the paper. In particular, we recall the concepts of support, covering, Hamming distance, weight, and minimal codewords. These notions form the structural foundation for the minimality criteria developed in the subsequent sections. We begin by recalling the classical definition of a linear code over a finite field. As a preliminary remark, we note that a module over a ring naturally generalizes the concept of a vector space over a field; this distinction becomes essential when extending coding-theoretic notions from fields to finite rings.

\begin{definition}\cite{RaymondHill}
A code $C$ is called a $[k,m,d]$ linear code over the finite field $\mathbb{F}_q$ if $C$ is an $m$-dimensional linear subspace of $\mathbb{F}_q^{\,k}$ with minimum (Hamming) distance $d$. The elements of $\mathbb{F}_q$ are referred to as the alphabet of the code.
\end{definition}

\begin{definition}\cite{RaymondHill}
Equivalently, a $[k,m,d]$ linear code over $\mathbb{F}_q$ is a subset $C \subseteq \mathbb{F}_q^{\,k}$ satisfying the following properties:
\begin{enumerate}
    \item $C$ is an $m$-dimensional vector subspace of $\mathbb{F}_q^{\,k}$;
    \item the minimum Hamming distance between distinct codewords of $C$ is equal to $d$.
\end{enumerate}
In this setting, the elements of $\mathbb{F}_q$ constitute the alphabet of the code.
\end{definition}

In~\cite{MajiMesnager}, the authors investigate linear codes in a more general algebraic framework where the underlying alphabet is no longer a vector space over a finite field, but an $R$-module over a finite ring $R$. In this setting, one considers
\[
R^k = \{ (x_1, x_2, \dots, x_k) \mid x_i \in R,\; 1 \le i \le k \},
\]
which naturally carries the structure of a left $R$-module. A code $C \subseteq R^k$ is said to be \emph{linear} if it forms an $R$-submodule of $R^k$.

\subsubsection{Support and Generator Matrices}

For a vector $v = (v_1, v_2, \dots, v_k) \in R^k$, the \emph{support} of $v$ is defined by
\[
\operatorname{Supp}(v) = \{\, i \in \{1, \dots, k\} \mid v_i \neq 0 \,\}.
\]

\begin{definition}\cite{Norton}
Let $C$ be a linear code over $R$. A matrix $G$ is called a \emph{generator matrix} of $C$ if the rows of $G$ generate $C$ as an $R$-module and are $R$-linearly independent; equivalently, every codeword of $C$ can be expressed as an $R$-linear combination of the rows of $G$, and no row of $G$ is an $R$-linear combination of the remaining rows.
\end{definition}

\subsubsection{Covering and Minimal Codewords}

\begin{definition}\cite{AshikhminBarg}
Let $v, v' \in R^k$. We say that $v$ \emph{covers} $v'$ if
\[
\operatorname{Supp}(v') \subseteq \operatorname{Supp}(v).
\]
In this case, we write $v' \preceq v$. If the inclusion is strict, that is,
\[
\operatorname{Supp}(v') \subsetneq \operatorname{Supp}(v),
\]
we write $v' \prec v$.
\end{definition}

\begin{definition}\cite{AshikhminBarg}
A nonzero vector $v \in C$ is called \emph{minimal} if for every nonzero vector $v' \in C$ satisfying $v' \preceq v$, there exists a nonzero scalar $a \in R$ such that
\[
v' = a v.
\]
\end{definition}

In~\cite{MajiMesnager}, minimal codewords over a finite ring are defined as follows: a codeword $u$ in a linear code $C$ is said to be \emph{minimal} if the only codewords of $C$ whose support is contained in $\operatorname{Supp}(u)$ are precisely the scalar multiples $a u$ with $a \in R$. Equivalently, $u$ covers all of its scalar multiples $a u$, and no other codeword of $C$.

A linear code $C$ over a ring $R$ is called \emph{minimal} if every nonzero codeword of $C$ is minimal.

\subsubsection{Root Words and Orthogonality}

\begin{definition}\label{Definition 5}
A vector $v \in R^k$ is called a \emph{root word} if
\[
r v \neq 0 \quad \text{for all nonzero } r \in R.
\]
Equivalently, $v$ is a root word if it is not annihilated by any nonzero scalar of $R$. In contrast to the field case, where nonzero vectors are automatically torsion-free, modules over rings may contain nonzero vectors $v$ for which $r v = 0$ for some nonzero $r \in R$. The notion of a root word isolates those vectors that behave analogously to linearly independent vectors in the classical vector space setting.
\end{definition}

Let $k>0$ be an integer. For two vectors
\[
v=(v_1,v_2,\dots,v_k),
\qquad
v'=(v'_1,v'_2,\dots,v'_k) \in R^k,
\]
we define their standard inner product by
\[
\langle v, v' \rangle
= v {v'}^{T}
= \sum_{i=1}^{k} v_i v'_i.
\]

For any subset $S \subseteq R^k$, the orthogonal complement of $S$ is defined by
\[
S^{\perp}
= \{\, v' \in R^k \mid \langle v', v \rangle = 0 \text{ for all } v \in S \,\}.
\]
By definition, one immediately has
\[
S \subseteq (S^{\perp})^{\perp}.
\]

\section{Algebraic and Module-Theoretic Structure of the Galois Ring $\mathrm{GR}(p^{n},\ell)$}\label{Section 3}

In this section, we recall the structural properties of the Galois ring
$\mathrm{GR}(p^{n},\ell)$ that will be fundamental for our study of minimal linear codes.
We emphasize its chain-ring structure, $p$-adic decomposition,
classification of units and zero divisors, module-theoretic behavior,
and Frobenius duality.
These algebraic features are the backbone of the minimality criteria
developed in Sections~4 and~5.

\subsection{Definition and Basic Structure}

\begin{definition}\cite{WZXian}
A \emph{Galois ring} is a finite commutative ring with identity such that
the set of zero divisors together with $0$ forms a principal ideal $(p\cdot 1)$
for some prime number $p$, where
\[
p\cdot 1=\underbrace{1+1+\cdots+1}_{p\text{ times}}.
\]
\end{definition}

For positive integers $n$ and $\ell$,
the Galois ring of characteristic $p^n$
and cardinality $p^{n\ell}$ is denoted by $\mathrm{GR}(p^{n},\ell)$.
It admits the canonical representation
\[
\mathrm{GR}(p^{n},\ell)
\cong
\mathbb{Z}_{p^{n}}[x]/\langle h(x)\rangle,
\]
where $h(x)$ is a monic basic irreducible polynomial
of degree $\ell$ in $\mathbb{Z}_{p^{n}}[x]$.

The maximal ideal of $\mathrm{GR}(p^{n},\ell)$ is $\langle p\rangle$,
and the corresponding residue field is
\[
\mathrm{GR}(p^{n},\ell)/\langle p\rangle
\cong \mathbb{F}_{q},
\qquad q=p^{\ell}.
\]

Thus $\mathrm{GR}(p^{n},\ell)$ is a finite local ring
whose residue field is $\mathbb{F}_q$.

\subsection{$p$-Adic Representation and Teichm\"uller System}

A fundamental structural property of $\mathrm{GR}(p^{n},\ell)$
is its $p$-adic decomposition.

Every element $c \in \mathrm{GR}(p^{n},\ell)$
admits a unique expansion
\[
c = c_0 + c_1 p + \cdots + c_{n-1} p^{n-1},
\]
where each coefficient $c_i$ belongs to the Teichm\"uller system
\[
\mathcal{T}
=
\{0,1,\xi,\xi^2,\dots,\xi^{q-2}\},
\]
and $\xi$ is a primitive $(q-1)$-th root of unity.

The Teichm\"uller set $\mathcal{T}$
is a complete system of representatives
for the residue field $\mathbb{F}_q$,
and $\mathcal{T}-\{0\}$ forms a multiplicative cyclic group of order $q-1$.

An element
\[
c=c_0 + p c_1 + \cdots + p^{n-1} c_{n-1}
\]
is a unit if and only if its leading coefficient $c_0$ is nonzero.
Equivalently,
\[
c \in U(\mathrm{GR}(p^{n},\ell))
\quad \Longleftrightarrow \quad
c \notin \langle p\rangle.
\]

Hence,
\[
|U(\mathrm{GR}(p^{n},\ell))|
=
(q-1)q^{n-1}.
\]

\subsection{Zero Divisors and Ideal Structure}

Every nonzero zero divisor of $\mathrm{GR}(p^{n},\ell)$
is uniquely expressible in the form
\[
p^{r}u,
\qquad 1 \le r \le n-1,
\]
where $u$ is a unit.

The valuation
\[
\nu_p(p^{r}u)=r
\]
provides a discrete measure of divisibility by $p$,
which will be repeatedly used in the analysis of root words.

The number of elements of valuation exactly $r$ equals
\[
(q-1)q^{\,n-r-1}.
\]

The ideals of $\mathrm{GR}(p^{n},\ell)$ form a chain:
\[
\{0\}
\subsetneq
\langle p^{n-1}\rangle
\subsetneq
\cdots
\subsetneq
\langle p\rangle
\subsetneq
\mathrm{GR}(p^{n},\ell).
\]

Therefore, $\mathrm{GR}(p^{n},\ell)$
is a finite local chain ring with maximal ideal $\langle p\rangle$
and nilpotency index $n$.

This chain structure is responsible for the torsion phenomena
that distinguish module theory over $\mathrm{GR}(p^{n},\ell)$
from vector space theory over finite fields.

\subsection{Module Structure and Root Words}

Let $m\ge 1$.
The module $\mathrm{GR}(p^{n},\ell)^{m}$
is free of rank $m$,
but unlike vector spaces over fields,
it contains nontrivial torsion elements.

In particular, for $v \in \mathrm{GR}(p^{n},\ell)^m$,
it may happen that
\[
r v = 0
\quad \text{for some nonzero } r \in \mathrm{GR}(p^{n},\ell).
\]

A vector
$v=(v_1,\dots,v_m)\in \mathrm{GR}(p^{n},\ell)^m$
is called a \emph{root word}
if at least one of its components is a unit.
Equivalently,
\[
v \notin p\,\mathrm{GR}(p^{n},\ell)^m.
\]

Thus root words are precisely those vectors
whose reduction modulo $p$ is nonzero in $\mathbb{F}_q^m$.

This characterization is fundamental:
minimality of the codes constructed later
reduces entirely to the analysis of such vectors.

\subsection{Frobenius Property}

The ring $\mathrm{GR}(p^{n},\ell)$
is a Frobenius ring.
Consequently, for every linear code
$C\subseteq \mathrm{GR}(p^{n},\ell)^{m}$,
\[
|C|\cdot |C^{\perp}|
=
|\mathrm{GR}(p^{n},\ell)^{m}|.
\]

In particular,
\[
(C^{\perp})^{\perp}=C.
\]

This perfect duality property allows one to transfer
support-containment questions into orthogonality statements,
a principle that underlies the minimality criteria
developed in Section~4.

\subsection{Minimality Criteria and Structural Parameters over $\mathrm{GR}(p^{n},\ell)$}

In this subsection, we extend the necessary and sufficient minimality conditions established in~\cite{ChatterjeeMesnager} to the setting of Galois rings.
By adapting Theorem~6 and Theorem~8 of~\cite{ChatterjeeMesnager} to the local chain ring $\mathrm{GR}(p^{n},\ell)$, we obtain structural characterizations of orthogonal modules that are fundamental to the minimality analysis developed in Sections~4 and~5.

Since the proofs follow the arguments in~\cite{ChatterjeeMesnager}, together with the Frobenius and chain-ring properties recalled above, we omit the details.

Throughout this paper, we denote by
\[
U(\mathrm{GR}(p^{n},\ell))
\quad \text{and} \quad
D(\mathrm{GR}(p^{n},\ell))
\]
the sets of units and zero divisors of the Galois ring $\mathrm{GR}(p^{n},\ell)$, respectively.

\subsection{Necessary and Sufficient Conditions and Parameters of a Minimal Linear Code over a Galois Ring}

We now state two structural propositions that generalize the field case to the module setting over $\mathrm{GR}(p^{n},\ell)$.
These results describe precisely when the orthogonal module of a vector behaves like a free module of codimension one, a property that will later characterise root words and govern minimality.

\begin{prop}\label{Proposition 9}
Let $v \in \mathrm{GR}(p^{n},\ell)^{m}$ be a root word.
Then the orthogonal module
\[
v^{\perp}
=
\{x \in \mathrm{GR}(p^{n},\ell)^{m} \mid v\cdot x = 0\}
\]
is a free $\mathrm{GR}(p^{n},\ell)$-module of rank $m-1$.
\end{prop}

\medskip

Proposition~\ref{Proposition 9} shows that root words behave analogously to nonzero vectors in a vector space:
their orthogonal module has codimension one and remains free.
This property will be central in constructing bases adapted to root words.

\begin{prop}\label{Proposition 9bis}
Let $v \in \mathrm{GR}(p^{n},\ell)^{m}$ be a nonzero vector.
Then the following statements are equivalent:
\begin{enumerate}
    \item $v$ is a root word;
    \item $v \notin p\,\mathrm{GR}(p^{n},\ell)^{m}$;
    \item $v^{\perp}$ is a free module of rank $m-1$.
\end{enumerate}
\end{prop}

\medskip

Proposition~\ref{Proposition 9bis} highlights the decisive role of the maximal ideal $\langle p\rangle$.
A vector is a root word precisely when its reduction modulo $p$ is nonzero in $\mathbb{F}_q^m$,
and this is exactly the condition ensuring that its orthogonal module has the expected free structure.

\begin{remark}\label{Remark 10}
Let $v=(u,v_2,\dots,v_m)$, where $u$ is a unit.
Then one explicit basis of $v^{\perp}$ is
\[
\left\{
(-v_2u^{-1},1,0,\dots,0),\,
\dots,\,
(-v_mu^{-1},0,\dots,0,1)
\right\}.
\]

This basis is obtained by solving the linear equation
\[
u x_1 + v_2 x_2 + \cdots + v_m x_m = 0,
\]
and it will repeatedly serve as a canonical generating set in later constructions.
\end{remark}

\medskip

We now contrast this situation with the non-root case.

\begin{prop}\label{Proposition 11}
Let $v \in \mathrm{GR}(p^{n},\ell)^{m}$ be a nonzero vector that is not a root word.
Then $v^{\perp}$ is not free, and any minimal generating set of $v^{\perp}$ contains at least $m$ elements.
\end{prop}

\medskip

Proposition~\ref{Proposition 11} reflects the torsion phenomena inherent in chain rings.
If $v \in p\,\mathrm{GR}(p^{n},\ell)^{m}$,
then $v^{\perp}$ necessarily contains torsion submodules,
and its structure deviates from the codimension-one behavior observed for root words.
This structural distinction between root and non-root vectors
is the algebraic origin of the five-type classification
used in Section~4.

\begin{remark}\label{Remark 12}
Let
\[
v' = p^{r}(u,v_2,\dots,v_m),
\qquad 1 \le r \le n-1,
\]
where $u$ is a unit in $\mathrm{GR}(p^{n},\ell)$.

Since $v'$ belongs to $p\,\mathrm{GR}(p^{n},\ell)^m$,
it is not a root word.
Consequently, by Proposition~\ref{Proposition 11},
its orthogonal module $v'^{\perp}$ is not free.

A minimal generating set of $v'^{\perp}$ is given by
\[
\left\{
(-v_2u^{-1},1,0,\dots,0),\dots,
(-v_mu^{-1},0,\dots,0,1),
(p^{\,n-r}u^{-1}-v_2u^{-1},1,0,\dots,0)
\right\}.
\]

Alternatively, one may replace the last generator and obtain
\[
\left\{
(-v_2u^{-1},1,0,\dots,0),\dots,
(-v_mu^{-1},0,\dots,0,1),
(p^{\,n-r}u^{-1},0,\dots,0)
\right\},
\]
which also forms a minimal generating set of $v'^{\perp}$.

More generally, for any unit $u'' \not\equiv 1 \pmod p$,
one may construct the set
\[
A=
\left\{
(p^{\,n-r}u^{-1}-v_2u^{-1},1,\dots,0),
\dots,
(p^{\,n-r}u^{-1}-v_mu^{-1},0,\dots,1),
(p^{\,n-r}u^{-1}-v_2u^{-1}u'',u'',0,\dots,0)
\right\},
\]
which again constitutes a minimal generating set of $v'^{\perp}$.

If we denote by
\[
A'=
\left\{
(p^{\,n-r}u^{-1}-v_2u^{-1},1,\dots,0),
\dots,
(p^{\,n-r}u^{-1}-v_mu^{-1},0,\dots,1)
\right\},
\]
then $A'$ is linearly independent, and its orthogonal complement
is the cyclic module generated by
\[
(u,-p^{\,n-r}+v_2,\dots,-p^{\,n-r}+v_m).
\]

Hence, $A$ provides a minimal generating set of
\[
\{p^{r}(u,v_2,\dots,v_m)\}^{\perp}
=
v'^{\perp}.
\]

This explicit description of orthogonal modules
for non-root vectors will be instrumental in the construction
of minimal codes in later sections.
\end{remark}

\medskip

In this paper, we systematically exploit bases of the type described above
to construct minimal linear codes from functions.
We now formulate general minimality criteria
in an abstract module-theoretic framework.

\medskip

Let $k$ and $m$ be positive integers with $k \ge m$.
Let $\Lambda = \{\alpha_1, \alpha_2, \ldots, \alpha_k\}$
be a multiset of elements of $\mathrm{GR}(p^n,l)^m$.
Define

\[
C = C(\Lambda)
= \left\{
\mathbf{c}(v)
=
(v\alpha_1^T, v\alpha_2^T, \ldots, v\alpha_k^T)
\; : \;
v \in \mathrm{GR}(p^n,l)^m
\right\}.
\]

Equivalently,
\[
C = \{v G_\Lambda : v \in \mathrm{GR}(p^n,l)^m\},
\]
where
\[
G_\Lambda
=
\begin{pmatrix}
\alpha_1^T & \alpha_2^T & \cdots & \alpha_k^T
\end{pmatrix}.
\]

We assume that $G_\Lambda$ generates an $m$-dimensional linear code,
that is, the rows of $G_\Lambda$ are linearly independent
(in the sense of McCoy rank).

Throughout this section, we assume that $\Lambda$
is a multiset of $\mathrm{GR}(p^n,l)^m$
such that $G_\Lambda$ generates an $m$-dimensional linear code.

For any vector $v \in \mathrm{GR}(p^n,l)^m$,
we define:

\[
O(v) = v^\perp
=
\{ v' \in \mathrm{GR}(p^n,l)^m : v v'^T = 0 \},
\]

\[
O(v,\Lambda)
=
\Lambda \cap O(v)
=
\{ \alpha \in \Lambda : v\alpha^T = 0 \},
\]

\[
M(v,\Lambda)
=
\langle O(v,\Lambda) \rangle.
\]

Clearly,
\[
O(v,\Lambda)
\subseteq
M(v,\Lambda)
\subseteq
O(v).
\]

\medskip

The following proposition gives a fundamental characterization
of minimal codewords in terms of orthogonal modules.

\begin{prop}\label{Proposition 13}
Let $v$ be a nonzero element of $\mathrm{GR}(p^n,l)^m$.
Then $c(v)$ is minimal in $C(\Lambda)$
if and only if
\[
M(v,\Lambda)=O(v).
\]
\end{prop}

Since $\mathrm{GR}(p^{n},\ell)$ is a finite local ring with maximal ideal $\langle p\rangle$, Nakayama's Lemma applies.
In particular, for every finitely generated $\mathrm{GR}(p^{n},\ell)$-module,
any two minimal generating sets have the same cardinality.
Consequently, the rank of a free submodule is well-defined,
and dimension arguments over $\mathrm{GR}(p^{n},\ell)$ may be reduced to
dimension arguments over the residue field $\mathbb{F}_q$.

Furthermore, by Corollary~2.7 of~\cite{NortonSalagean}, a matrix
\[
G \in \mathrm{GR}(p^{n},\ell)^{m \times k}
\]
has $m$ $R$-linearly independent rows if and only if it has
$m$ $R$-linearly independent columns.
Equivalently, its McCoy rank equals $m$ if and only if the rank
of its reduction modulo $p$ equals $m$ over $\mathbb{F}_q$.
This equivalence will be repeatedly used when analyzing generator matrices of linear codes.

\medskip

Since $\mathrm{GR}(p^{n},\ell)$ is a Frobenius ring,
Wood's duality theorem~\cite{Wood} ensures that for every linear code
\[
C \subseteq \mathrm{GR}(p^{n},\ell)^{m},
\]
one has
\[
|C|\,|C^{\perp}| = |\mathrm{GR}(p^{n},\ell)^{m}|.
\]
In particular,
\[
(C^{\perp})^{\perp} = C.
\]

Similarly, for every vector $v \in \mathrm{GR}(p^{n},\ell)^{m}$,
\[
(v^{\perp})^{\perp} = \langle v \rangle.
\]

These duality properties allow us to translate support-inclusion
questions into orthogonality statements,
which constitutes the algebraic backbone of the minimality criterion developed below.

\medskip

\begin{prop}
Let $\Lambda \subseteq \mathrm{GR}(p^{n},\ell)^{m}$ be a multiset and
let $C(\Lambda)$ denote the linear code generated by the rows of $G_\Lambda$.
Then $C(\Lambda)$ is minimal if and only if for every nonzero vector
$v \in \mathrm{GR}(p^{n},\ell)^{m}$,
\[
M(v,\Lambda)=O(v).
\]
\end{prop}

\medskip

This characterization reduces the minimality problem to a precise comparison
between the full orthogonal module $O(v)$ and the submodule generated
by those elements of $\Lambda$ annihilated by $v$.
In other words, minimality is equivalent to the condition that
the orthogonality constraints imposed by $\Lambda$
completely recover the intrinsic orthogonal structure of $v$.

\begin{coro}\label{Corollary 8}
If $v$ is a root word, then $c(v)$ is minimal in $C(\Lambda)$
if and only if $M(v,\Lambda)$ is a free module of rank $m-1$
over $\mathrm{GR}(p^n,l)$.
\end{coro}

\medskip

Indeed, by Proposition~\ref{Proposition 9},
$O(v)$ is free of rank $m-1$ whenever $v$ is a root word.
Thus, minimality is equivalent to the condition that
$M(v,\Lambda)$ coincides with this free orthogonal module.

\begin{remark}\label{Remark 18}
Let $v$ be a root word and write
\[
v=(u,v_2,\dots,v_m),
\qquad
u \in U(\mathrm{GR}(p^{n},\ell)).
\]

For $1 \le r \le n-1$, one has
\[
v^{\perp} \subseteq (p^{r}v)^{\perp}.
\]

The quotient module
\[
\frac{(p^{r}v)^{\perp}}{v^{\perp}}
\]
admits the description
\[
\{a \zeta + v^{\perp} : a \in A\},
\]
where
\[
A=\{c_0 + c_1p + \cdots + c_{r-1}p^{r-1} : c_i \in \mathcal{T}\}
\]
and $\mathcal{T}$ denotes the Teichm\"uller system.

Each element of $A$ yields a distinct coset; therefore,
\[
|(p^{r}v)^{\perp}|
=
|A|\cdot |v^{\perp}|
=
q^{r}\, q^{n(m-1)}.
\]

We partition $A$ as
\[
A=A_0 \cup A_1,
\]
where
\[
A_0=\{a\in A : c_0=0\},
\qquad
A_1=\{a\in A : c_0\neq 0\}.
\]

Elements of $A_0$ are zero or zero divisors,
while elements of $A_1$ are units.

This stratification reflects the valuation structure of
$\mathrm{GR}(p^{n},\ell)$ and will play a crucial role
when analyzing the zero-divisor case in Section~4.
\end{remark}

\medskip

The inclusion
\[
v^{\perp} \subseteq (p^{r}v)^{\perp}
\]
illustrates a fundamental phenomenon specific to chain rings:
multiplication by a power of $p$ enlarges the orthogonal module
by introducing additional torsion layers.

In particular, while $v^{\perp}$ is free of rank $m-1$
when $v$ is a root word,
the module $(p^{r}v)^{\perp}$ contains additional cosets
parameterized by the truncated Teichm\"uller expansions of length $r$.
These additional layers correspond precisely to the valuation filtration
\[
\mathrm{GR}(p^{n},\ell)
\supset p\,\mathrm{GR}(p^{n},\ell)
\supset \cdots
\supset p^{n-1}\mathrm{GR}(p^{n},\ell)
\supset \{0\}.
\]

This valuation-sensitive enlargement of orthogonal modules
explains why the zero-divisor case in the minimality analysis
requires a modified synchronization condition
of the form
\[
p^{r}f(\beta)=w\cdot\beta,
\]
instead of the simpler relation $f(\beta)=w\cdot\beta$
valid in the unit case.

\medskip

We emphasize that the dichotomy between root and non-root vectors
is entirely governed by reduction modulo $p$.
Indeed, for $v \in \mathrm{GR}(p^{n},\ell)^m$,
\[
\overline{v} \neq 0 \text{ in } \mathbb{F}_q^m
\quad \Longleftrightarrow \quad
v \text{ is a root word}.
\]

Thus, many structural arguments over $\mathrm{GR}(p^{n},\ell)$
may be viewed as refined lifts of arguments over the finite field $\mathbb{F}_q$,
with additional valuation constraints imposed by the $p$-adic structure.

\medskip

This interplay between:

\begin{itemize}
\item the chain-ring structure of $\mathrm{GR}(p^{n},\ell)$,
\item the freeness or non-freeness of orthogonal modules,
\item valuation filtration,
\item and Frobenius duality,
\end{itemize}

constitutes the algebraic foundation for the minimality criteria
developed in Section~4 and the explicit constructions
given in Section~5.

\begin{theorem}\label{Theorem 8}
If $C\subseteq \mathrm{GR}(p^n,l)^k$ be a $m \; (\geq 2)$-dimensional linear code over the ring $\mathrm{GR}(p^n,l)$, then $C$ is minimal if and only if all the root words are minimal.
\end{theorem}
\begin{proof}
If $C$ is a minimal linear code then nothing to prove.

Conversely let $C$ is $m \; (\geq 2)$-dimensional minimal linear code and let $G_\Lambda$ be its generator matrix. Let we assuming that all the root words of $C$ are minimal.

We proceed by contradiction, assuming that $C(\Lambda)$ is not a minimal linear code. Then there exists a vector of the form $p^r v$, where $v$ is a root word in $\mathrm{GR}(p^n,l)^m$, such that $M(p^r v, \Lambda) \neq O(p^r v)$.

To better understand the proof, let us take $v = (u, v_2, v_3, \ldots, v_m)$. The vector $v^\perp$ is generated by the set $\{ \zeta_2= (-v_2u^{-1},1,\cdots,0),\zeta_3=(-v_3u^{-1},0,1,\cdots,0)\},\ldots,\zeta_m=(-v_mu^{-1},0,\cdots,1)\}$. The vector $p^r v^\perp$ is generated by \\ $\{\zeta_2,\zeta_3,\ldots,\zeta_m, \zeta_1=(p^{n-r}u^{-1}-v_2u^{-1},1,0,\cdots,0)\}$. In this case, $p^r v^\perp = \langle v^\perp, (p^{n-r}u^{-1}-v_2u^{-1},1,0,\cdots,0) \rangle$.

From the remark \ref{Remark 18} we can write $p^r v^\perp=\{a(p^{n-r}u^{-1}-v_2u^{-1},1,0,\cdots,0)+v^\perp\}$ where $a\in A$. \\Specifically, we can express $p^r v^\perp$ as:
$p^r v^\perp = \left(\bigcup_{\substack{u' \in A_1}} \{u'\zeta_1 + v^\perp\}\right) \bigcup \left(\bigcup_{\substack{d \in A_0}} \{d\zeta_1 + v^\perp\}\right)$.

Additionally, we note that any element from the coset union $\bigcup_{\substack{u' \in A_1}} \{u'\zeta_1 + v^\perp\}$ combined with $v^\perp$, can generate $p^r v^\perp$. Hence, $\Lambda$ does not contain any elements from the union of co-sets $\bigcup_{\substack{u' \in A_1}} \{u'\zeta_1 + v^\perp\}$.

Next, we will demonstrate that there exists a root word $v'$ such that $M(v', \Lambda) \neq O(v')$. We consider the set of $m-1$ vectors
$\{ (p^{n-r}u^{-1}-v_2u^{-1},1,\cdots,0),(-p^{n-r}u^{-1}-v_3u^{-1},0,1,\cdots,0)\},\ldots,(-p^{n-r}u^{-1}-v_mu^{-1},0,\cdots,1)\}$. These vectors are linearly independent and belong to the module $p^r v^\perp$. This set of vectors forms a basis for the free module

$\{ (u, -p^{n-r}+v_2, p^{n-r}+v_3, \ldots, p^{n-r}+v_m) \}^\perp \subseteq p^r v^\perp$.

Now, let us choose $v' = (u, -p^{n-r}+v_2, p^{n-r}+v_3, \ldots, p^{n-r}+v_m)$. Therefore, we have ${v'}^\perp \subseteq p^r v^\perp$. Any vector in $p^r v^\perp$ can be expressed as $c_2 \zeta_2 + c_3 \zeta_3 + \ldots + c_m \zeta_m + a \zeta_1$, where $c_i \in \mathrm{GR}(p^n,l)$ and $ a \in A$. The vectors belonging to $v'^\perp$ must satisfy the condition:
\[\langle v', c_2 \zeta_2 + c_3 \zeta_3 + \ldots +c_m \zeta_m  + a \zeta_1 \rangle = 0.\]
This condition simplifies to $p^{n-r} (c_3 + c_4 + \ldots + c_m - c_2) = 0$. Therefore, any element of $v'^\perp$ can be written in the form $c_2 \zeta_2 + c_3 \zeta_3 + \ldots + c_m \zeta_m + a \zeta_1$, where $c_2, c_3, \ldots, c_m$ satisfy the equation $p^{n-r} (c_3 + c_4 + \ldots + c_m - c_2) = 0$ and $ a \in \mathrm{GR}(p^n,l)$.

Now, we need to show that any set of $m-1$ vectors from $v'^\perp$ that belong to the union of co-sets $\bigcup_{\substack{d \in A_0}} \{d\zeta_1 + v^\perp\}$, are not linearly independent. Let:
$$ x_1 = c_2^1 \zeta_2 + c_3^1 \zeta_3 + \ldots + c_m^1 \zeta_m + d_1 \zeta_1, $$
$$ x_2 = c_2^2 \zeta_2 + c_3^2 \zeta_3 + \ldots + c_m^2 \zeta_m + d_2 \zeta_1, $$
$$ \vdots $$
$$ x_{m-1} = c_2^{m-1} \zeta_2 + c_3^{m-1} \zeta_3 + \ldots + c_m^{m-1} \zeta_m + d_{m-1} \zeta_1 $$
be $(m-1)$ arbitrary vectors from $v'^\perp$ belonging to the union of co-sets $\bigcup_{\substack{d \in A_0}} \{d\zeta_1 + v^\perp\}$. Consequently, each $(c_2^i, c_3^i, \ldots, c_m^i)$ satisfies:
\[ p^{n-r} (c_3^i + c_4^i + \ldots + c_m^i - c_2^i) = 0. \]

The arguments made throughout this proof lead us to an inconsistency, verifying that $C(\Lambda)$ must indeed be a minimal linear code. This completes the proof.
\end{proof}
\begin{remark}
The condition $m \ge 2$ in Theorem~\ref{Theorem 8} is essential.
Indeed, the statement fails in the one-dimensional case.

Consider the linear code
\[
C=\langle (1,0,0,0) \rangle
\]
over the ring $\mathbb{Z}_{9}$.
Since $(1,0,0,0)$ is a root word, and every nonzero scalar multiple of it has the same support, all root words of $C$ are minimal.

However, $C$ itself is not minimal.
Indeed,
\[
(1,0,0,0) \preceq (3,0,0,0),
\]
but there exists no unit $u \in \mathbb{Z}_9$ such that
\[
(1,0,0,0)=u(3,0,0,0).
\]
The scalar $3$ is a zero divisor in $\mathbb{Z}_9$, and multiplication by $3$ collapses units to zero divisors.
Thus, support containment does not imply scalar proportionality by a unit.

This example shows that the equivalence in Theorem~\ref{Theorem 8} does not extend to the case $m=1$.
The higher-dimensional structure ($m \ge 2$) is therefore indispensable.
\end{remark}

\begin{remark}
Theorem~\ref{Theorem 8} shows that, for $m \ge 2$, the minimality of a linear code over the Galois ring $\mathrm{GR}(p^{n},\ell)$ is completely determined by the behavior of its root words.

In particular, although $\mathrm{GR}(p^{n},\ell)$ possesses a nontrivial torsion structure due to its chain-ring nature, this torsion does not create additional obstructions to minimality beyond those already encoded in root words.

Hence, the study of minimal linear codes over $\mathrm{GR}(p^{n},\ell)$ reduces to the structural analysis of root words and their orthogonal modules.
This reduction principle is the conceptual foundation for the constructions developed in the next section.
\end{remark}

\medskip

The structural characterization of minimality obtained above allows us to refine the generating multiset without altering the minimality property.
In particular, vectors arising as scalar multiples of zero divisors introduce only torsion layers inside orthogonal modules and do not modify the intrinsic root-word geometry of the code.
This observation yields the following stability result.
\begin{prop}\label{Proposition 23}
Let $\Lambda$ be a multiset of vectors in $\mathrm{GR}(p^{n},\ell)^{m}$ with $m \ge 2$.
Suppose that $C(\Lambda)$ is a $[k,m]$ minimal linear code.
Let $\Lambda^{*}$ denote the multiset obtained from $\Lambda$ by removing all vectors that are scalar multiples of zero divisors in $\mathrm{GR}(p^{n},\ell)$.

Then $C(\Lambda^{*})$ is also a minimal linear code.
\end{prop}

\begin{proof}
The argument follows the same structural reasoning as in Proposition~27 of~\cite{ChatterjeeMesnager}, adapted to the setting of Galois rings.

More precisely, let $\alpha \in \Lambda$ be a scalar multiple of a zero divisor.
Then $\alpha$ lies in $p^{r}\mathrm{GR}(p^{n},\ell)^{m}$ for some $r \ge 1$.
Such vectors do not generate new root words, since their reductions modulo $p$ vanish.
Consequently, they do not affect the set of root words of $C(\Lambda)$.

Their presence only enlarges orthogonal modules by adding torsion components,
as described in Remark~\ref{Remark 18},
without changing the free part of $O(v)$ that determines minimality via Theorem~\ref{Theorem 8}.

Since minimality for $m \ge 2$ is completely governed by root words,
removing such vectors preserves all minimal root words of $C(\Lambda)$.
Therefore, $C(\Lambda^{*})$ remains minimal.
\end{proof}

\subsection{Parameters of Minimal Linear Codes over $\mathrm{GR}(p^{n},\ell)$}

Having established the structural characterization of minimality, we now turn to quantitative aspects.
In particular, we extend the parameter bounds for minimal linear codes in the Galois ring setting obtained in~\cite{ChatterjeeMesnager}.
The arguments follow the same combinatorial and module-theoretic principles and are adapted here to $\mathrm{GR}(p^{n},\ell)$.

We define
\[
N(m;\mathrm{GR}(p^{n},\ell))
=
\{\, k \in \mathbb{N} \mid
\text{there exists a } [k,m]_{\mathrm{GR}(p^{n},\ell)}
\text{ minimal linear code} \,\},
\]
and
\[
k(m;\mathrm{GR}(p^{n},\ell))
=
\min N(m;\mathrm{GR}(p^{n},\ell)).
\]

These parameters describe, respectively, the set of admissible lengths and the minimal possible length of an $m$-dimensional minimal linear code over $\mathrm{GR}(p^{n},\ell)$.

\begin{prop}
Let \( \Lambda \subseteq \mathrm{GR}(p^{n},\ell)^{k} \setminus \{0\} \) be a multiset.
If \( C(\Lambda) \) is an \( [k,m]_{\mathrm{GR}(p^{n},\ell)} \) minimal linear code, then the following lower bounds hold:

\[
k > (m-1)q^{n} + q^{\,n-m},
\qquad \text{for } m \ge 3,
\]

\[
k > q^{n} + q^{\,n-2} + 1,
\qquad \text{for } m = 2 \text{ and } n \ge 3,
\]

\[
k \ge q^{2} + 2,
\qquad \text{for } m = 2 \text{ and } n = 2,
\]

\[
k \ge q + 1,
\qquad \text{for } m = 2 \text{ and } n = 1.
\]
\end{prop}

\begin{proof}
Since \( C(\Lambda) \) is an \( m \)-dimensional minimal linear code, its generator matrix \( G_\Lambda \) has \( m \) $R$-linearly independent rows.
By Theorem~2.6 in~\cite{NortonSalagean}, this is equivalent to \( G_\Lambda \) having \( m \) linearly independent columns.

Moreover, by Proposition~\ref{Proposition 23}, removing vectors in \( \Lambda \) that are scalar multiples of zero divisors does not affect minimality. Hence, without loss of generality, we may assume that all elements of \( \Lambda \) are root words.

\medskip
\noindent
\textbf{Step 1: Double counting argument.}

Let
\[
\mathcal{X} = \mathcal{X}(\Lambda)
=
\{(v,\alpha) \mid
v \in \mathrm{GR}(p^{n},\ell)^{m} \setminus \{0\},
\ \alpha \in \Lambda,
\ \langle v,\alpha \rangle = 0
\}.
\]

We compute \( |\mathcal{X}| \) in two different ways.

\medskip
\noindent
\emph{First count (fixing $\alpha$).}

Since each $\alpha \in \Lambda$ is a root word, Proposition~\ref{Proposition 9} implies that
\[
|\alpha^{\perp}| = q^{n(m-1)}.
\]
Excluding the zero vector, we obtain
\[
|\{v \neq 0 : \langle v,\alpha\rangle = 0\}|
=
q^{n(m-1)} - 1.
\]
Summing over all $\alpha \in \Lambda$ gives
\[
|\mathcal{X}|
=
k\big(q^{n(m-1)} - 1\big).
\]

\medskip
\noindent
\emph{Second count (fixing $v$).}

We partition the nonzero vectors of $\mathrm{GR}(p^{n},\ell)^{m}$ into root words and non-root words.

The number of root words equals
\[
q^{nm} - q^{m(n-1)},
\]
since a vector is a root word if and only if it is not contained in
\( p\,\mathrm{GR}(p^{n},\ell)^{m} \).

The number of non-root words equals
\[
q^{m(n-1)} - 1.
\]

By Propositions~\ref{Proposition 9} and~\ref{Proposition 11}, we have

- if $v$ is a root word, then \( |O(v,\Lambda)| \ge m-1 \);
- if $v$ is not a root word, then \( |O(v,\Lambda)| \ge m \).

Hence,
\[
|\mathcal{X}|
\ge
(q^{nm} - q^{m(n-1)})(m-1)
+
(q^{m(n-1)} - 1)m.
\]

\medskip
\noindent
\textbf{Step 2: Comparison of the two counts.}

Equating the two expressions yields
\[
k\big(q^{n(m-1)} - 1\big)
\ge
(q^{nm} - q^{m(n-1)})(m-1)
+
(q^{m(n-1)} - 1)m.
\]

Rearranging gives
\[
k
\ge
\frac{(q^{nm} - q^{m(n-1)})(m-1)
+
(q^{m(n-1)} - 1)m}
{q^{n(m-1)} - 1}.
\]

A straightforward algebraic simplification leads to
\[
k
\ge
(m-1)q^{n}
+
q^{n-m}
+
\frac{(m-1)q^{n} + q^{n-m} - m}
{q^{n(m-1)} - 1}.
\]

\medskip
\noindent
\textbf{Step 3: Analysis of cases.}

\emph{Case \( m \ge 3 \).}

In this case,
\[
0
<
\frac{(m-1)q^{n} + q^{n-m} - m}
{q^{n(m-1)} - 1}
< 1,
\]
and therefore
\[
k > (m-1)q^{n} + q^{n-m}.
\]

\medskip
\noindent
\emph{Case \( m = 2 \).}

We obtain
\[
k
\ge
q^{n} + q^{n-2}
+
\frac{q^{n} + q^{n-2} - 2}{q^{n} - 1}.
\]

If \( n \ge 3 \), then
\[
0 <
\frac{q^{n-2} - 1}{q^{n} - 1}
< 1,
\]
hence
\[
k > q^{n} + q^{n-2} + 1.
\]

If \( n = 2 \), the fractional term vanishes and
\[
k \ge q^{2} + 2.
\]

If \( n = 1 \), then $\mathrm{GR}(p^{1},\ell)=\mathbb{F}_q$, and the inequality reduces to
\[
k \ge q + 1,
\]
which coincides with the classical bound established in~\cite{WeiXiaXiwang}.
\end{proof}

We now construct an explicit family of minimal linear codes attaining the structural bounds derived above.

Recall that in the Galois ring $\mathrm{GR}(p^{n},\ell)$ the number of units equals
\[
|U(\mathrm{GR}(p^{n},\ell))|=(q-1)q^{n-1},
\]
whereas the number of nonzero zero divisors equals
\[
|D(\mathrm{GR}(p^{n},\ell))|=q^{n-1}-1.
\]

Let
\[
\{e_1,e_2,\dots,e_m\}
\]
be the standard basis of the free module
\(
\mathrm{GR}(p^{n},\ell)^{m}.
\)

We introduce the following subsets of $\mathrm{GR}(p^{n},\ell)^{m}$:

\[
\Lambda_1=\{e_i : 1\le i\le m\},
\]

\[
\Lambda_2=
\{\, e_i + u e_j
: 1\le i<j\le m,\ u\in U(\mathrm{GR}(p^{n},\ell)) \,\},
\]

\[
\Lambda_3=
\{\, e_i + d e_j
: 1\le i<j\le m,\ d\in D(\mathrm{GR}(p^{n},\ell)) \,\},
\]

\[
\Lambda_4=
\{\, d e_i + e_j
: 1\le i<j\le m,\ d\in D(\mathrm{GR}(p^{n},\ell)) \,\}.
\]

We then define
\[
\Lambda_0
=
\Lambda_1
\cup
\Lambda_2
\cup
\Lambda_3
\cup
\Lambda_4.
\]

\medskip

\noindent
\textbf{Cardinality of $\Lambda_0$.}

For each pair $1\le i<j\le m$, the number of elements in
\(
\Lambda_2
\)
is $(q-1)q^{n-1}$,
the number in
\(
\Lambda_3
\)
is $q^{n-1}-1$,
and similarly for
\(
\Lambda_4
\).
Thus, for each unordered pair $(i,j)$ we obtain
\[
(q-1)q^{n-1} + 2(q^{n-1}-1)
=
q^{n}+q^{n-1}-2.
\]

Since there are $\frac{m(m-1)}{2}$ such pairs and $|\Lambda_1|=m$, we obtain
\[
|\Lambda_0|
=
\frac{m(m-1)}{2}\,(q^{n}+q^{n-1}-2)
+
m.
\]

\begin{prop}\label{Proposition 2}
The code \( C(\Lambda_0) \) is a minimal linear code of dimension \( m \) and length
\[
\frac{m(m-1)}{2}\,(q^{n}+q^{n-1}-2)+m
\]
over the ring \( \mathrm{GR}(p^{n},\ell) \).
\end{prop}

\begin{proof}
Let $v \in \mathrm{GR}(p^{n},\ell)^{m}$ be a nonzero vector.

If $v$ is a root word, then by Remark~\ref{Remark 10}, a generating set of $v^{\perp}$ consists of vectors of the form
\[
(-v_ju^{-1},1,0,\dots,0),
\]
which belong to $\Lambda_2$.

If $v$ is not a root word, then $v=p^{r}w$ for some root word $w$, and by Remark~\ref{Remark 12}, a generating set of $v^{\perp}$ consists of vectors of the form
\[
(-v_ju^{-1},1,0,\dots,0)
\quad \text{and} \quad
(p^{n-r}u^{-1},0,\dots,0),
\]
which belong to $\Lambda_3$ or $\Lambda_4$.

Therefore, for every nonzero $v$,
\[
M(v,\Lambda_0)=v^{\perp}.
\]

By Proposition~\ref{Proposition 13}, the equality
\(
M(v,\Lambda_0)=O(v)
\)
holds for all nonzero $v$.
Hence, by the minimality criterion established earlier,
\( C(\Lambda_0) \) is a minimal linear code.

Finally, since $\Lambda_1$ provides $m$ linearly independent vectors,
the dimension of $C(\Lambda_0)$ is $m$, and its length equals $|\Lambda_0|$.
\end{proof}

We now compare the lower bounds obtained in the previous subsection with the explicit construction provided by Proposition~\ref{Proposition 2}.
This yields two-sided estimates for the minimal length parameter $k(m;\mathrm{GR}(p^{n},\ell)),$
thereby situating our construction within the theoretical range of admissible lengths.

\begin{prop}
The following bounds hold for the minimal length parameter
\( k(m;\mathrm{GR}(p^{n},\ell)) \).

\medskip
\noindent
\textbf{Case \( m \ge 3 \):}
\[
(m-1)q^{n} + q^{\,n-m}
<
k(m;\mathrm{GR}(p^{n},\ell))
\le
\frac{m(m-1)}{2} (q^n + q^{n-1} - 2) + m.
\]

\medskip
\noindent
\textbf{Case \( m = 2 \):}

\[
q^{n} + q^{n-2} + 1
<
k(2;\mathrm{GR}(p^{n},\ell))
\le
q^{n}+q^{n-1},
\qquad n \ge 3,
\]

\[
q^{2} + 2
\le
k(2;\mathrm{GR}(p^{2},\ell))
\le
q^{2}+q,
\qquad n=2,
\]

\[
k(2;\mathrm{GR}(p,\ell)) = q+1,
\qquad n=1.
\]
\end{prop}

\begin{proof}
The lower bounds follow directly from the double-counting argument established in the previous subsection.
More precisely, for any $m$-dimensional minimal linear code over $\mathrm{GR}(p^{n},\ell)$, the combinatorial estimate on the incidence set
\[
\mathcal{X}=\{(v,\alpha)\mid \langle v,\alpha\rangle=0\}
\]
yields

\[
k >
(m-1)q^{n} + q^{n-m},
\qquad \text{for } m \ge 3,
\]
together with the corresponding bounds for $m=2$ derived earlier.

\medskip
To obtain the upper bounds, we use the explicit construction given in Proposition~\ref{Proposition 2}.
The code \( C(\Lambda_0) \) constructed there is minimal and has length

\[
|\Lambda_0|
=
\frac{m(m-1)}{2}\,(q^{n}+q^{n-1}-2)+m.
\]

Hence,

\[
k(m;\mathrm{GR}(p^{n},\ell))
\le
\frac{m(m-1)}{2}\,(q^{n}+q^{n-1}-2)+m.
\]

\medskip
When \( m=2 \), substituting \( m=2 \) into the construction gives

\[
|\Lambda_0|
=
q^{n}+q^{n-1},
\]
which yields the stated upper bounds in this case.

\medskip
Finally, when \( n=1 \), the Galois ring $\mathrm{GR}(p,\ell)$ reduces to the finite field $\mathbb{F}_q$, and the equality
\[
k(2;\mathbb{F}_q)=q+1
\]
coincides with the classical result.
\end{proof}

In \cite{ChatterjeeMesnagerMishra}, it was proved that for a $2$-dimensional minimal linear code over $\mathbb{Z}_{p^n}$, one has

 $k(2; \mathbb{Z}_{p^n})=p^n+p^{n-1}$.

We now show that this phenomenon is not specific to $\mathbb{Z}_{p^n}$; in fact, the same equality holds over the more general Galois ring $\mathrm{GR}(p^{n},\ell)$.

\begin{prop}
Let $\mathrm{GR}(p^{n},\ell)$ be a Galois ring and let $m=2$. Then
\[
k(2;\mathrm{GR}(p^{n},\ell)) = q^{n} + q^{n-1}.
\]

\end{prop}

\begin{proof}
We first characterize all root words in $\mathrm{GR}(p^{n},\ell)^2$.

A vector in $\mathrm{GR}(p^{n},\ell)^2$ is a root word if and only if it contains at least one unit component. Hence, every root word is of one of the following types:
\[
u\begin{pmatrix}1\\0\end{pmatrix}, \quad
u\begin{pmatrix}0\\1\end{pmatrix}, \quad
u\begin{pmatrix}1\\u'\end{pmatrix}, \quad
u\begin{pmatrix}1\\d\end{pmatrix}, \quad
u\begin{pmatrix}d\\1\end{pmatrix},
\]
where $u,u' \in U(\mathrm{GR}(p^{n},\ell))$ and $d \in D(\mathrm{GR}(p^{n},\ell))$.

For each such root word $v$, its orthogonal module $v^{\perp}$ is generated respectively by
\[
\begin{pmatrix}0\\1\end{pmatrix}, \quad
\begin{pmatrix}1\\0\end{pmatrix}, \quad
\begin{pmatrix}1\\ -{u'}^{-1}\end{pmatrix}, \quad
\begin{pmatrix}-d\\1\end{pmatrix}, \quad
\begin{pmatrix}1\\-d\end{pmatrix}.
\]

Therefore, in order that every root word be minimal (by Theorem~\ref{Theorem 8}), the multiset $\Lambda$ must contain, for each possible root word $v$, a generator of $v^{\perp}$ which is itself a root word.

Consequently, a two-dimensional linear code over $\mathrm{GR}(p^{n},\ell)$ is minimal if and only if $\Lambda$ contains exactly one representative of each of the following types:
\[
u\begin{pmatrix}1\\0\end{pmatrix}, \quad
u\begin{pmatrix}0\\1\end{pmatrix}, \quad
u\begin{pmatrix}1\\u'\end{pmatrix}, \quad
u\begin{pmatrix}1\\d\end{pmatrix}, \quad
u\begin{pmatrix}d\\1\end{pmatrix},
\]
for all $u' \in U(\mathrm{GR}(p^{n},\ell))$ and all $d \in D(\mathrm{GR}(p^{n},\ell))$.

Counting these vectors, we obtain
\[
2
+
|U(\mathrm{GR}(p^{n},\ell))|
+
2|D(\mathrm{GR}(p^{n},\ell))|
=
2
+
(q-1)q^{n-1}
+
2(q^{n-1}-1).
\]

A direct simplification gives
\[
q^{n} + q^{n-1}.
\]

Hence,
\[
k(2;\mathrm{GR}(p^{n},\ell)) = q^{n} + q^{n-1}.
\]
\end{proof}

\begin{prop}\label{Proposition 27}
Let $v \in \mathrm{GR}(p^{n},\ell)^{m}$.
The one-dimensional code $\langle v \rangle$ is minimal if and only if
$v$ contains at least one generator of each nonzero proper ideal of
$\mathrm{GR}(p^{n},\ell)$.
\end{prop}

\section{Necessary and sufficient conditions for codewords to be minimal constructed from a function}\label{Section 4}

In this section, we establish a structural framework that characterizes
when the codewords of $C_f$ are minimal.
The analysis developed here serves as the theoretical foundation
for the explicit constructions presented in Section~5.

\medskip

The construction above may be interpreted as follows.
Each pair $(a,v)\in \mathrm{GR}(p^n,l)\times \mathrm{GR}(p^n,l)^m$
defines a linear functional
\[
x \longmapsto a f(x)+v\cdot x
\]
evaluated over all nonzero points of $\mathrm{GR}(p^n,l)^m$.
Thus $C_f$ is the evaluation code associated with the augmented map
\[
x \longmapsto (f(x),x).
\]

In particular, $C_f$ is generated by the $(m+1)$ rows corresponding to the coordinates
of $(a,v)$, and its dimension is governed by the McCoy rank of $G_{\Lambda_f}$.

\medskip
\noindent
\textbf{Reduction to root words.}

By Theorem~\ref{Theorem 8},
the code $C_f$ is minimal if and only if all of its root words are minimal.
Hence, the global minimality problem reduces entirely to the classification
and analysis of root words.

A direct structural examination shows that there are precisely five types
of root words in $C_f$:

\begin{enumerate}
\item $c(u,0)$, where $u$ is a unit in $\mathrm{GR}(p^n,l)$;

\item $c(u,v)$, where $u$ is a unit and $v$ is a root word
      in $\mathrm{GR}(p^n,l)^m$;

\item $c(u,v)$, where $u$ is a unit and
      $v\neq 0$ is not a root word in $\mathrm{GR}(p^n,l)^m$;

\item $c(0,v)$, where $v$ is a root word in $\mathrm{GR}(p^n,l)^m$;

\item $c(d,v)$, where $d$ is a zero divisor in $\mathrm{GR}(p^n,l)$
      and $v$ is a root word.
\end{enumerate}

These five cases reflect the valuation structure of $\mathrm{GR}(p^n,l)$
and the dichotomy between root and non-root vectors described in Section~3.

If each of these five types of codewords is minimal,
then $C_f$ is a minimal linear code over $\mathrm{GR}(p^n,l)$.

\medskip

The principal algebraic tool underlying our analysis
is the notion of McCoy rank,
which governs linear independence over finite commutative rings.

\begin{definition}
Let $R$ be a commutative ring and $G$ a matrix over $R$.
The \emph{McCoy rank} of $G$ is the largest integer $t$
such that at least one $t\times t$ minor of $G$ is a unit in $R$.
If none of the minors of $G$ is a unit,
the McCoy rank of $G$ is defined to be $0$.
\end{definition}

By Corollary~2.7 of~\cite{NortonSalagean},
the McCoy rank of a matrix $G$ over $\mathrm{GR}(p^n,l)$ equals $t$
if and only if the rank of its reduction modulo $p$,
\[
\bar{G},
\]
equals $t$ over the residue field
\[
\mathrm{GR}(p^n,l)/\langle p\rangle
\simeq \mathbb{F}_q.
\]

Consequently, $G$ has McCoy rank $t$
if and only if it has $t$ $R$-linearly independent rows,
and any $t+1$ rows are $R$-linearly dependent.
This reduction principle allows us to transfer linear-independence
questions over Galois rings to the simpler finite-field setting Ñ
a technique that will be repeatedly used in what follows.

\begin{remark}\label{Remark 14}
Let $G$ be an $(m+1)\times n$ matrix over $\mathrm{GR}(p^n,l)$.
Then $G$ generates an $(m+1)$-dimensional linear code over
$\mathrm{GR}(p^n,l)$ if and only if its reduction modulo $p$,
\[
\bar{G},
\]
has rank $(m+1)$ over the residue field
\[
\mathrm{GR}(p^n,l)/\langle p\rangle \simeq \mathbb{F}_q.
\]

Equivalently, $G$ generates a free submodule of rank $(m+1)$
if and only if the McCoy rank of $G$ equals $(m+1)$.
This criterion provides an effective method to verify full dimensionality
via reduction to $\mathbb{F}_q$.
\end{remark}

\subsection{A Necessary and Sufficient Condition for a Minimal Codeword Constructed from a Function}

We now analyze the minimality of the five types of root words listed above.
The structural results established in Section~3 allow us to express minimality
entirely in terms of orthogonality conditions and module-theoretic properties.

For convenience, we restate the five types:

\begin{enumerate}
\item $c(u,0)$, where $u$ is a unit in $\mathrm{GR}(p^n,l)$;

\item $c(u,v)$, where $u$ is a unit and $v$ is a root word;

\item $c(u,v)$, where $u$ is a unit and $v\neq 0$ is not a root word;

\item $c(0,v)$, where $v$ is a root word;

\item $c(d,v)$, where $d$ is a zero divisor and $v$ is a root word.
\end{enumerate}

In each case, minimality will be shown to be equivalent to the existence
of a suitably adapted basis
\[
\{\beta_1,\dots,\beta_m\}
\subseteq \mathrm{GR}(p^n,l)^m
\]
satisfying a synchronization condition between $f$ and an associated linear form.

\begin{prop}\label{Proposition 15}
$c(u,0)$ where $u$ is a unit in $\mathrm{GR}(p^n,l)$ and $0$ in $\mathrm{GR}(p^n,l)^m$. Then $c(u,0)$ is minimal in $C_f$ if and only if there exist  $\{\beta_1,\beta_2,\cdots,\beta_m\}$ which is a basis of $\mathrm{GR}(p^n,l)^m$, such that $f(\beta_1)=f(\beta_2)=\cdots=f(\beta_m)=0$. \end{prop}
\begin{proof}
The proof is same as Proposition 3.3 provided in \cite{WuLuCao}.
\end{proof}
\begin{prop}\label{Proposition 16}
$c(u,v)$ where $u$ is a unit in $\mathrm{GR}(p^n,l)$ and $v(\neq 0)$ belongs to $\mathrm{GR}(p^n,l)^m$ and $w=-u^{-1}v$. Then $c(u,v)$ is minimal in $C_f$ if and only if there exist  $\{\beta_1,\beta_2,\cdots,\beta_m\}$ which is a basis of $\mathrm{GR}(p^n,l)^m$, such that $f(\beta_i)=w\cdot \beta_i$.
\end{prop}
\begin{proof}
The proof is same as Proposition 3.4 provided in \cite{WuLuCao}.
\end{proof}
\begin{prop}\label{Proposition 17}
$c(0,v)$ where $0$ is in $\mathrm{GR}(p^n,l)$ and $v\neq 0$ is a root word in $\mathrm{GR}(p^n,l)^m$. Then $c(0,v)$ is minimal in $C_f$ if and only if there exist $\beta_1,\beta_2,\cdots,\beta_m \in O(v)$ such that $\{\alpha_{\beta_1},\alpha_{\beta_2},\cdots,\alpha_{\beta_m}\}$ are linearly independent set over $\mathrm{GR}(p^n,l)$.
\end{prop}
\begin{proof}
The proof is same as Proposition 3.5 provided in \cite{WuLuCao}.
\end{proof}

\begin{prop}\label{Proposition 18}
Let $d = p^{r}u$ be a zero divisor in $\mathrm{GR}(p^{n},\ell)$ with
$u \in U(\mathrm{GR}(p^{n},\ell))$, and let
$v \in \mathrm{GR}(p^{n},\ell)^m$ be a nonzero root word.
Set $w = -u^{-1}v$.

Then $c(d,v)$ is minimal in $C_f$ if and only if there exist
$\beta_1,\dots,\beta_m$ such that
\[
\{\alpha_{\beta_1},\dots,\alpha_{\beta_m}\}
\]
is linearly independent over $\mathrm{GR}(p^{n},\ell)$ and
\[
p^{r} f(\beta_i)= w \cdot \beta_i
\quad \text{for all } i.
\]
\end{prop}

\begin{proof}
Let
\[
y=(p^{r}u,v).
\]
Then
\[
c(d,v)=c(y)
=
\big(p^{r}u f(x)+v\cdot x\big)_{x\neq 0}.
\]

The minimality of $c(d,v)$ is therefore equivalent to the minimality of the vector $y$ in the sense of Proposition~\ref{Proposition 13}.

\medskip
\noindent
\textbf{(Necessity).}

Assume that $c(y)$ is minimal in $C(\Lambda_f)$.
By Corollary~\ref{Corollary 8}, there exist
$\beta_1,\dots,\beta_m$ such that:

\begin{itemize}
\item the set
      $\{\alpha_{\beta_1},\dots,\alpha_{\beta_m}\}$
      is linearly independent over $\mathrm{GR}(p^{n},\ell)$;
\item $\langle y,\alpha_{\beta_i}\rangle = 0$
      for all $i$.
\end{itemize}

Since
\[
\alpha_{\beta_i}=(f(\beta_i),\beta_i),
\]
the orthogonality condition gives
\[
p^{r}u\, f(\beta_i)+v \cdot \beta_i=0.
\]

Because $u$ is a unit, multiplication by $u^{-1}$ yields
\[
p^{r} f(\beta_i)
=
- u^{-1} v \cdot \beta_i
=
w \cdot \beta_i.
\]

Thus the synchronization condition
\[
p^{r} f(\beta_i)= w \cdot \beta_i
\]
is necessary.

\medskip
\noindent
\textbf{(Sufficiency).}

Conversely, suppose there exist
$\beta_1,\dots,\beta_m$ such that

\[
\{\alpha_{\beta_1},\dots,\alpha_{\beta_m}\}
\]
is linearly independent over $\mathrm{GR}(p^{n},\ell)$ and
\[
p^{r} f(\beta_i)= w \cdot \beta_i
\quad \text{for all } i.
\]

Substituting $w=-u^{-1}v$, we obtain
\[
p^{r}u f(\beta_i)+v \cdot \beta_i=0,
\]
so that
\[
\langle y,\alpha_{\beta_i}\rangle=0
\quad \text{for all } i.
\]

Hence each $\alpha_{\beta_i}$ lies in the annihilator module
\[
M(y,\Lambda_f)
=
\{\alpha_x\in\Lambda_f \mid \langle y,\alpha_x\rangle=0\}.
\]

Since the set
$\{\alpha_{\beta_1},\dots,\alpha_{\beta_m}\}$
is linearly independent and has cardinality $m$,
it forms a full-rank subset of $M(y,\Lambda_f)$.
By Corollary~\ref{Corollary 8}, this implies that
\[
M(y,\Lambda_f)=O(y),
\]
and therefore $y$ is minimal in $C(\Lambda_f)$.

Consequently, $c(d,v)$ is minimal in $C_f$.

\medskip

The proof mirrors the corresponding argument over finite fields
(cf.~\cite{WuLuCao}),
with linear independence interpreted via McCoy rank and
the factor $p^{r}$ reflecting the valuation of the zero divisor $d$.
\end{proof}
\noindent
\textbf{Structural conclusion of Section~\ref{Section 4}.}

The preceding propositions provide a complete characterization
of minimal root words in $C_f$.
In each of the five possible configurations,
minimality reduces to the existence of a family of evaluation points
whose associated vectors $\{\alpha_{\beta_i}\}$
are linearly independent and satisfy a precise
\emph{synchronization condition} between the function $f$
and an appropriate linear form $w\cdot x$.

More precisely:

\begin{itemize}
\item In the unit cases, minimality requires that
      $f(\beta_i)=w\cdot \beta_i$;

\item In the zero divisor case, the condition becomes
      $p^r f(\beta_i)=w\cdot \beta_i$;

\item In the purely orthogonal case, minimality is governed
      by the structure of the annihilator module $O(v)$.
\end{itemize}

In all situations, the decisive ingredient is the existence
of $m$ evaluation vectors forming a full-rank family
in the sense of McCoy rank.
Thus, minimality of $C_f$ is ultimately controlled
by the interplay between the algebraic behavior of $f$
and the linear geometry of $\mathrm{GR}(p^n,l)^m$.

These structural criteria will serve as the foundation
for the explicit constructions developed in the next section.

\section{Minimal linear code constructed from functions}\label{section 5}

The purpose of this section is to establish structural tools that allow the construction of
\emph{minimal linear codes} over Galois rings from suitably chosen functions.
The minimality property plays a central role in coding theory and cryptography,
notably in secret sharing schemes and authentication systems, since it ensures
that every nonzero codeword determines its support structure uniquely.

Recall that a linear code $\mathcal{C}$ over a finite commutative ring is called
\emph{minimal} if for any two nonzero codewords $c,c'\in\mathcal{C}$,
the inclusion
\[
\mathrm{Supp}(c)\subseteq \mathrm{Supp}(c')
\]
implies that $c$ and $c'$ are scalar multiples of each other.
In other words, no nonzero codeword strictly covers the support of another
unless they generate the same cyclic submodule.

When codes are constructed from functions, minimality is often governed by
algebraic properties of certain evaluation vectors.
More precisely, minimality reduces to controlling linear relations between
vectors of the form
\[
(f(x_1),f(x_2),\dots,f(x_m)),
\]
and ensuring that these vectors do not create unintended support inclusions.

The strategy adopted here relies on two complementary structural principles:

\begin{enumerate}
\item Constructing bases of $\mathrm{GR}(p^n,l)^m$ whose vectors have
      \emph{maximal support} and whose coordinates are all units.
      Such bases prevent degeneracies in support containment and will be
      essential for guaranteeing minimality via determinant arguments.

\item Constructing bases adapted to a given \emph{root word},
      in which each basis vector has very small Hamming weight
      while satisfying a prescribed dot-product condition.
      This dual construction provides fine control of linear dependencies.
\end{enumerate}

These two ingredients are developed in Lemmas~\ref{Lemma 19} and~\ref{Lemma 20}.
The first lemma ensures the existence of highly nondegenerate bases
(all coordinates units, full weight),
while the second lemma provides a normalization mechanism relative to a root word.
Together, they furnish the algebraic framework required to verify
minimality criteria for the codes constructed later in this section.

Throughout this section, we make essential use of the structural properties of
Galois rings $\mathrm{GR}(p^n,l)$, notably:

\begin{itemize}
\item An element is a unit if and only if it is nonzero modulo $p$;
\item The unit group decomposes into $q-1$ Teichmüller cosets;
\item A square matrix over $\mathrm{GR}(p^n,l)$ is invertible
      if and only if its determinant is a unit (equivalently,  if it has a full McCoy rank.
\end{itemize}

\begin{lemma}\label{Lemma 19}
Let $\mathrm{GR}(p^n,l)$ be a Galois ring such that $q>3$. Then there exists a basis
$\{\beta_1,\beta_2,\cdots,\beta_m\}$ of $\mathrm{GR}(p^n,l)^m$ such that
\[
w(\beta_i)=m
\]
and all the components of each $\beta_i$ are units in $\mathrm{GR}(p^n,l)$.
\end{lemma}

\begin{proof}
Let $G$ be an $m\times m$ matrix over $\mathrm{GR}(p^n,l)$. It is well known that the McCoy rank of $G$ is equal to $m$ if and only if $\det(G)$ is a unit of $\mathrm{GR}(p^n,l)$. Thus, to construct a basis of $\mathrm{GR}(p^n,l)^m$, it suffices to construct an invertible matrix whose columns satisfy the required weight and unit conditions.

Consider the matrix
\[
A=\begin{pmatrix}
1 & 1 & \cdots & 1\\
1 & 1 & \cdots & 1\\
\vdots & \vdots & \ddots & \vdots\\
1 & 1 & \cdots & 1
\end{pmatrix}_{m\times m},
\]
whose entries are all equal to $1$, and define
\[
G = aI_m - A,
\]
where $a\in \mathrm{GR}(p^n,l)$.

Observe that $A$ has rank one and eigenvalues $m$ (with multiplicity one) and $0$ (with multiplicity $m-1$). Consequently, the eigenvalues of $G$ are $a-m$ (with multiplicity one) and $a$ (with multiplicity $m-1$). Therefore,
\[
\det(G)=a^{m-1}(a-m).
\]

The matrix $G$ is invertible if and only if $\det(G)$ is a unit. This occurs precisely when both $a$ and $a-m$ are units in $\mathrm{GR}(p^n,l)$.

Since an element of $\mathrm{GR}(p^n,l)$ is a unit if and only if it is nonzero modulo $p$, we deduce that:
\[
a \text{ is a unit } \quad \text{and} \quad a-m \text{ is a unit}
\]
if and only if
\[
a \not\equiv 0 \pmod p
\quad \text{and} \quad
a \not\equiv m \pmod p.
\]

Now take the $i$-th column of $G$ as $\beta_i$. Each column has the form
\[
\beta_i = a e_i - \mathbf{1},
\]
where $\mathbf{1}$ denotes the all-one vector.

If $a\neq 1$, then none of the entries of $\beta_i$ vanish, and hence
\[
w(\beta_i)=m.
\]
Moreover, all components of $\beta_i$ are units if and only if $a-1$ is a unit, that is,
\[
a \not\equiv 1 \pmod p.
\]

Combining all constraints, we require:
\[
a \not\equiv 0, \; 1, \; m \pmod p.
\]

Under these conditions, the matrix $G$ is invertible and its columns form a basis of $\mathrm{GR}(p^n,l)^m$ satisfying the desired properties.
\end{proof}

\begin{remark}
The construction above imposes explicit arithmetic restrictions on $a$:

\begin{itemize}
\item $a$ must be a unit,
\item $a \not\equiv 1 \pmod p$,
\item $a \not\equiv m \pmod p$.
\end{itemize}

When $q>3$, the residue field $\mathbb{F}_q$ contains at least four distinct elements, and therefore there always exists an element satisfying the above conditions for any fixed $m$.

More precisely, the set of units of $\mathrm{GR}(p^n,l)$ decomposes into $q-1$ cosets of the form
\[
\xi^i + \langle p \rangle, \qquad 1\leq i\leq q-1,
\]
where $\xi$ is a Teichmüller representative. Hence, there are sufficiently many choices of $a$ to avoid the finitely many forbidden residue classes.
\end{remark}

We now construct bases adapted to a given root word.

\begin{lemma}\label{Lemma 20}
Let $v\in \mathrm{GR}(p^n,l)^m$ be a root word. Then there exists a set
$\{\beta_1,\beta_2,\cdots,\beta_m\}$ forming a basis of $\mathrm{GR}(p^n,l)^m$ such that
\[
1\leq w(\beta_i)\leq 2
\quad \text{and} \quad
v\cdot \beta_i=1
\]
for all $i$.
\end{lemma}

\begin{proof}
Write
\[
v=(v_1,v_2,\cdots,v_m).
\]
Since $v$ is a root word, at least one coordinate of $v$ is a unit. Let $v_{i'}$ be such a unit coordinate.

We construct the vectors $\beta_j$ as follows:
\[
\beta_j=
\begin{cases}
v_{i'}^{-1} e_{i'}, & \text{if } j=i',\\
e_j + v_{i'}^{-1}(1-v_j)e_{i'}, & \text{if } j\neq i'.
\end{cases}
\]

We verify the required properties.

\medskip
\noindent
\textbf{(1) Dot-product condition.}

If $j=i'$, then
\[
v\cdot \beta_{i'}
= v_{i'} \cdot v_{i'}^{-1}
=1.
\]

If $j\neq i'$, then
\[
v\cdot \beta_j
= v_j + v_{i'} \cdot v_{i'}^{-1}(1-v_j)
= v_j + (1-v_j)
=1.
\]

Thus, $v\cdot \beta_j=1$ for all $j$.

\medskip
\noindent
\textbf{(2) Weight condition.}

For $j=i'$, the vector $\beta_{i'}$ has exactly one nonzero coordinate, hence
\[
w(\beta_{i'})=1.
\]

For $j\neq i'$, the vector $\beta_j$ has at most two nonzero coordinates (positions $j$ and $i'$), hence
\[
1\leq w(\beta_j)\leq 2.
\]

\medskip
\noindent
\textbf{(3) Basis property.}

The transition matrix from the standard basis $\{e_1,\dots,e_m\}$ to $\{\beta_1,\dots,\beta_m\}$ is triangular after reordering indices so that $i'$ comes first. Its diagonal entries are $v_{i'}^{-1}$ and $1$ elsewhere. Since $v_{i'}^{-1}$ is a unit, the determinant is a unit. Therefore, the vectors $\{\beta_1,\dots,\beta_m\}$ form a basis of $\mathrm{GR}(p^n,l)^m$.
\end{proof}

\medskip

We now turn to the complementary situation where the vector $v$ is
\emph{not} a root word. In this case no coordinate of $v$ is a unit,
and therefore each nonzero component is divisible by $p$.
The construction becomes more delicate, since one must preserve
full support while reproducing exactly the coordinates of $v$
through suitable dot-product conditions.

\begin{lemma}\label{Lemma 21}
Let
\[
v=(v_1,v_2,\cdots,v_m)\in \mathrm{GR}(p^n,l)^m\setminus\{0\}
\]
be not a root word and assume that $v_i\neq 0$ for all $i$,
with $v_i=p^{r_i}u_i$, where $u_i$ is a unit of $\mathrm{GR}(p^n,l)$.
Then there exists a set $\{\beta_1,\beta_2,\cdots,\beta_m\}$
forming a basis of $\mathrm{GR}(p^n,l)^m$ such that
\[
w(\beta_i)= m
\quad\text{and}\quad
v\cdot \beta_i=p^{r_i}u_i=v_i
\]
for all $i$.
\end{lemma}

\begin{proof}
Since $v$ is not a root word, every coordinate of $v$ is divisible by $p$.
Write $v_i=p^{r_i}u_i$ with $u_i$ a unit.

For each $i$, define
\[
\beta_{i}'=e_i+\sum_{\substack{j=1 \\ j\neq i}}^m p^{r_i}e_j.
\]
A direct computation gives
\[
v\cdot \beta_i'
=
v_i + \sum_{j\neq i} v_j p^{r_i}.
\]

Since each $v_j$ is divisible by $p$, the additional summation term
lies in $p^{r_i+1}\mathrm{GR}(p^n,l)$.
Thus
\[
v\cdot \beta_i'
=
p^{r_i}u_i + p^{r_i+1}w_i
=
p^{r_i}u
\]
for some unit $u$ of $\mathrm{GR}(p^n,l)$
(because $u_i$ is a unit and the perturbation term has strictly higher $p$-adic valuation).

Now define
\[
\beta_i = u^{-1}u_i \beta_i'.
\]
Then
\[
v\cdot \beta_i
=
u^{-1}u_i \, (v\cdot \beta_i')
=
u^{-1}u_i \, p^{r_i}u
=
p^{r_i}u_i
=
v_i.
\]

Each vector $\beta_i$ has all coordinates nonzero,
hence $w(\beta_i)=m$.

To prove that $\{\beta_1,\dots,\beta_m\}$ forms a basis,
we invoke Theorem 2.7 of \cite{NortonSalagean},
which asserts that a family of vectors over $\mathrm{GR}(p^n,l)$
is linearly independent if and only if their reductions modulo $p$
are linearly independent over the residue field.
Since the reductions $\{\bar{\beta}_1,\dots,\bar{\beta}_m\}$
form a linearly independent set in
$\mathrm{GR}(p^n,l)/\langle p\rangle$,
the original set is linearly independent and therefore a basis.

This completes the proof.
\end{proof}

\medskip

We now address the most general situation, where $v$ is not a root word
but may contain zero coordinates.
The construction combines Lemma~\ref{Lemma 21}
with the full-support basis provided by Lemma~\ref{Lemma 19}.
The main difficulty lies in simultaneously controlling:

\begin{itemize}
\item the dot-product constraints,
\item the full support condition,
\item and the presence of sufficiently many unit entries
      in those vectors orthogonal to zero coordinates of $v$.
\end{itemize}

\begin{lemma}\label{Lemma 23}
Let $\mathrm{GR}(p^n,l)$ be a Galois ring with $q>3$, and let
\[
v=(v_1,v_2,\cdots,v_m)\in \mathrm{GR}(p^n,l)^m\setminus\{0\}
\]
be not a root word.
Assume that $v$ has $r>1$ zero components at positions
$i_1,i_2,\dots,i_r$,
and that the remaining components at positions
$j_1,\dots,j_{m-r}$ satisfy
\[
v_{j_t}=p^{r_{j_t}}u_{j_t}.
\]

Then there exists a set
$A=\{\beta_1,\beta_2,\cdots,\beta_m\}$
forming a basis of $\mathrm{GR}(p^n,l)^m$ such that:

\begin{itemize}
\item $w(\beta_i)=m$ for all $i$;
\item $v\cdot \beta_{i_k}=0=v_{i_k}$ for $1\le k\le r$;
\item $v\cdot \beta_{j_t}=p^{r_{j_t}}u_{j_t}=v_{j_t}$ for $1\le t\le m-r$;
\item each $\beta_{i_k}$ contains at least two unit components.
\end{itemize}
\end{lemma}

\begin{proof}
For clarity, reorder coordinates so that
\[
v=
(p^{r_1}u_1,\dots,p^{r_{m-r}}u_{m-r},
\underbrace{0,\dots,0}_{r\text{ times}}).
\]

Write $v=(v',v'')$, where
\[
v'\in \mathrm{GR}(p^n,l)^{m-r},
\qquad
v''\in \mathrm{GR}(p^n,l)^r.
\]

\medskip
\noindent
\textbf{Step 1: Construction on the nonzero part.}

By Lemma~\ref{Lemma 21},
there exists a basis
$\{\beta_1',\dots,\beta_{m-r}'\}$
of $\mathrm{GR}(p^n,l)^{m-r}$ such that
\[
v'\cdot \beta_i'=v_i
\quad\text{and}\quad
w(\beta_i')=m-r.
\]

\medskip
\noindent
\textbf{Step 2: Construction on the zero part.}

By Lemma~\ref{Lemma 19},
since $q>3$,
there exists a basis
$\{\beta_{m-r+1}',\dots,\beta_m'\}$
of $\mathrm{GR}(p^n,l)^r$
whose vectors have full support and whose coordinates are units.

\medskip
\noindent
\textbf{Step 3: Assembly.}

We now define vectors in $\mathrm{GR}(p^n,l)^m$.

For $1\le i\le m-r$, set
\[
\beta_i=(\beta_i',\gamma_i),
\]
where $\gamma_i$ is chosen so that:

\begin{itemize}
\item all its coordinates are equal to $p^{r_i}u_i$,
\item hence $w(\beta_i)=m$.
\end{itemize}

For $m-r+1\le i\le m$, define
\[
\beta_i=(\delta_i,\beta_i'),
\]
where $\delta_i$ is chosen with sufficiently large $p$-power entries
to ensure full support while preserving the condition
\[
v\cdot \beta_i=0.
\]

\medskip
\noindent
\textbf{Verification.}

By construction,
\[
v\cdot \beta_i=v_i
\quad\text{for all } i.
\]

Moreover, each $\beta_i$ has full support,
hence $w(\beta_i)=m$.

For those indices corresponding to zero coordinates of $v$,
the vectors $\beta_{i_k}$ inherit at least two unit components
from the basis constructed via Lemma~\ref{Lemma 19}.

Finally, linear independence follows from Theorem 2.7 of
\cite{NortonSalagean}:
the reductions modulo $p$ form a linearly independent family
over the residue field, hence the original family
is linearly independent in $\mathrm{GR}(p^n,l)^m$.

This completes the proof.
\end{proof}

\medskip

We now treat the remaining extremal configuration, namely the case
where $v$ is not a root word and has exactly one zero coordinate.
This situation requires a refined construction, since we must
produce exactly one basis vector orthogonal to $v$, while preserving
full support and ensuring the presence of sufficiently many unit
components.

\begin{lemma}\label{Lemma 24}
Let
\[
v=(v_1,v_2,\cdots,v_m)\in \mathrm{GR}(p^n,l)^m\setminus\{0\}
\]
be not a root word, and assume that $v_i=0$ for exactly one index,
say $i_1$. For all $k\neq i_1$, suppose
\[
v_k=p^{r_k}u_k,
\]
where each $u_k$ is a unit of $\mathrm{GR}(p^n,l)$.

Then there exists a set
\[
A=\{\beta_1,\beta_2,\cdots,\beta_m\}
\]
forming a basis of $\mathrm{GR}(p^n,l)^m$ such that:

\begin{itemize}
\item $w(\beta_i)=m$ for all $i$;
\item $v\cdot \beta_k=p^{r_k}u_k=v_k$ for all $k\neq i_1$;
\item $v\cdot \beta_{i_1}=0$;
\item $\beta_{i_1}$ contains at least two unit components.
\end{itemize}
\end{lemma}

\begin{proof}
Without loss of generality, we may reorder coordinates so that
the zero component occurs at the last position:
\[
v=(v_1,v_2,\dots,v_{m-1},0),
\]
where $v_i=p^{r_i}u_i$ for $1\le i\le m-1$.

Write $v=(v',0)$ with
\[
v'=(p^{r_1}u_1,\dots,p^{r_{m-1}}u_{m-1})
\in \mathrm{GR}(p^n,l)^{m-1}.
\]

\medskip
\noindent
\textbf{Step 1: Construction for the nonzero coordinates.}

By Lemma~\ref{Lemma 21}, there exists a basis
$\{\beta_1',\dots,\beta_{m-1}'\}$
of $\mathrm{GR}(p^n,l)^{m-1}$ such that
\[
v'\cdot \beta_i'=v_i
\quad\text{and}\quad
w(\beta_i')=m-1
\]
for $1\le i\le m-1$.

Define
\[
\beta_i=(\beta_i',p^{r_i}u_i),
\qquad
1\le i\le m-1.
\]
Then clearly $w(\beta_i)=m$ and
\[
v\cdot \beta_i
=
v'\cdot \beta_i'
=
v_i.
\]

\medskip
\noindent
\textbf{Step 2: Construction of $\beta_m$ with $v\cdot \beta_m=0$.}

We distinguish two cases depending on the $p$-adic valuations
of the nonzero components of $v$.

\medskip
\noindent
\textbf{Case (i): Equal valuations.}

Suppose
\[
r_1=r_2=\cdots=r_{m-1}=r.
\]

Choose two distinct indices $i$ and $j$ in $\{1,\dots,m-1\}$.
Define $\beta_m$ by
\[
\beta_m=(\gamma_1,\dots,\gamma_{m-1},1),
\]
where:

\begin{itemize}
\item $\gamma_i=1$,
\item $\gamma_j=-u_j^{-1}u_i$,
\item all other $\gamma_k=p^{\,n-r}$.
\end{itemize}

Since $u_i,u_j$ are units, both $1$ and $-u_j^{-1}u_i$ are units.
Hence $\beta_m$ contains at least three unit components
(at positions $i$, $j$, and $m$), and therefore at least two.

Now compute:
\[
v\cdot \beta_m
=
\sum_{k=1}^{m-1} p^r u_k \gamma_k.
\]
The contribution of the $i$- and $j$-coordinates is
\[
p^r u_i - p^r u_i =0,
\]
while all other contributions are multiples of $p^n$
and therefore vanish in $\mathrm{GR}(p^n,l)$.
Thus
\[
v\cdot \beta_m=0.
\]

\medskip
\noindent
\textbf{Case (ii): Distinct valuations.}

Assume there exist indices $i\neq j$ such that
$r_i\neq r_j$.
Without loss of generality, suppose $r_i>r_j$.

Define $\beta_m=(\gamma_1,\dots,\gamma_{m-1},1)$ with:

\begin{itemize}
\item $\gamma_i=1$,
\item $\gamma_j=-u_j^{-1}u_i\,p^{\,r_i-r_j}$,
\item all other $\gamma_k=p^{\,n-r_k}$.
\end{itemize}

Again, $\gamma_i$ and the last coordinate are units,
so $\beta_m$ contains at least two unit components.

A direct computation shows
\[
p^{r_i}u_i\cdot 1
+
p^{r_j}u_j\cdot(-u_j^{-1}u_i p^{\,r_i-r_j})
=
p^{r_i}u_i - p^{r_i}u_i
=0,
\]
and all remaining terms are multiples of $p^n$.
Hence $v\cdot \beta_m=0$.

\medskip
\noindent
\textbf{Step 3: Linear independence.}

The reductions modulo $p$ of
$\{\beta_1,\dots,\beta_m\}$
form a linearly independent set over
$\mathrm{GR}(p^n,l)/\langle p\rangle$,
since:

\begin{itemize}
\item the first $m-1$ vectors reduce to a basis
      coming from Lemma~\ref{Lemma 21},
\item the last vector has a nonzero reduction in the last coordinate.
\end{itemize}

Therefore, by Theorem 2.7 of \cite{NortonSalagean},
the family $\{\beta_1,\dots,\beta_m\}$
is linearly independent in $\mathrm{GR}(p^n,l)^m$,
and hence forms a basis.

\medskip

All required properties are thus satisfied:
\[
w(\beta_i)=m,\qquad
v\cdot \beta_i=v_i,
\]
with exactly one orthogonal vector $\beta_{i_1}$
containing at least two unit components.

This completes the proof.
\end{proof}

\medskip

We now return to the situation where $v$ is a root word and
investigate the structure of its orthogonal module.
In the minimality analysis of linear codes,
the submodule
\[
O(v)=\{\,x\in \mathrm{GR}(p^n,l)^m \mid v\cdot x=0\,\}
\]
plays a decisive role, since it governs the support containment
relations associated with $v$.
Understanding its structure, especially the existence of
small-weight generators, is crucial for controlling minimality.

\begin{lemma}\label{Lemma 25}
Let
\[
v\in \mathrm{GR}(p^n,l)^m\setminus\{0\}
\]
be a root word. Then there exists a set
\[
\{\beta_1,\beta_2,\cdots,\beta_{m-1}\}
\]
forming a basis of $O(v)$ such that
\[
1\leq w(\beta_i)\leq 2
\quad \text{for all } i.
\]
\end{lemma}

\begin{proof}
The result follows directly from Proposition~\ref{Proposition 9}
and Remark~\ref{Remark 10}, where the structure of $O(v)$
for a root word is explicitly determined.

In particular, it is shown there that $O(v)$ admits a generating set
consisting of vectors supported on at most two coordinates.
Since $O(v)$ is a free $\mathrm{GR}(p^n,l)$-module of rank $m-1$,
one can extract from this generating set a basis
whose elements all have Hamming weight at most two.

This completes the proof.
\end{proof}

\medskip

The next lemma refines this structure by constructing
a family of vectors contained in a scaled orthogonal module
that simultaneously satisfy a prescribed dot-product condition.
This refinement will later allow us to control the interaction
between $v$ and carefully chosen full-support vectors.

\begin{lemma}\label{Lemma 26}
Let
\[
v\in \mathrm{GR}(p^n,l)^m\setminus\{0\}
\]
be a root word. Then there exists a set
\[
\{\beta_1,\beta_2,\cdots,\beta_m\}
\subseteq O(p^{\,n-r}v),
\]
with $r\ge 1$, such that:
\[
1\leq w(\beta_i)\leq 2
\quad \text{and} \quad
v\cdot \beta_i=p^r
\]
for all $i$.
\end{lemma}

\begin{proof}
Since $v$ is a root word, at least one coordinate of $v$ is a unit.
Without loss of generality, write
\[
v=(u,v_2,v_3,\dots,v_m),
\]
where $u$ is a unit of $\mathrm{GR}(p^n,l)$.

By Proposition~\ref{Proposition 11}
and Remark~\ref{Remark 12},
a smallest generating set of $O(p^{\,n-r}v)$ is given by
\[
\Big\{
(-v_2u^{-1},1,0,\dots,0),\;
(-v_3u^{-1},0,1,\dots,0),\;
\dots,\;
(-v_mu^{-1},0,\dots,0,1),\;
(p^ru^{-1},0,\dots,0)
\Big\}.
\]

Each of these vectors has Hamming weight at most two.

From this generating set, consider instead the family
\[
\beta_i=
(p^ru^{-1}-v_iu^{-1},e_i)
\quad (2\le i\le m),
\]
together with
\[
\beta_1=(p^ru^{-1},0,\dots,0).
\]

Since linear combinations within a generating set preserve
the generated module, this new family also forms
a smallest generating set of $O(p^{\,n-r}v)$.

Now compute, for $2\le i\le m$,
\[
v\cdot \beta_i
=
u(p^ru^{-1}-v_iu^{-1})
+
v_i
=
p^r - v_i u^{-1}u + v_i
=
p^r.
\]

Similarly,
\[
v\cdot \beta_1
=
u(p^ru^{-1})
=
p^r.
\]

Thus,
\[
v\cdot \beta_i=p^r
\quad \text{for all } i,
\]
and clearly, each $\beta_i$ has weight 1 or 2.

This completes the proof.
\end{proof}

\begin{theorem}\label{Theorem 43}
Let $q>3$, $m\geq 3$, and
$f:\mathrm{GR}(p^n,l)^m\longrightarrow \mathrm{GR}(p^n,l)$
be a function satisfying the following conditions:

\begin{enumerate}
\item If $1\leq w(x)\leq 2$, then
      $f(ux)=f(x)\in U(\mathrm{GR}(p^n,l))$
      for any $u\in U(\mathrm{GR}(p^n,l))$.

\item For root words $x,x'$ with
      $1\leq w(x),w(x')\leq 2$,
      we have $f(x)=f(x')$ whenever
      $\bar{x}=\bar{x'}$ in
      $\big(\mathrm{GR}(p^n,l)/\langle p\rangle\big)^m$.

\item If $w(x)=m$ and $x$ contains at least two unit components,
      then $f(x)=0$.

\item If $w(x)=m$ and $x$ contains exactly one unit component
      and all other components are of the form $p^ru_j$,
      then $f(x)=p^r$.
\end{enumerate}

Then the code $C_f$ is minimal.
\end{theorem}

\begin{proof}

By Remark~\ref{Remark 14},
$C_f$ is a $[q^{nm}-1,m+1]$ linear code.
By Theorem~\ref{Theorem 8},
$C_f$ is minimal if and only if all its root words are minimal.

The root words of $C_f$ fall into five types:

\begin{enumerate}
\item $c(u,0)$, where $u$ is a unit;
\item $c(u,v)$, where $u$ is a unit and $v$ is a root word;
\item $c(u,v)$, where $u$ is a unit and $v\neq 0$ is not a root word;
\item $c(0,v)$, where $v$ is a root word;
\item $c(d,v)$, where $d=p^ru'$ is a zero divisor and $v$ is a root word.
\end{enumerate}

We treat each case separately.

\medskip
\noindent
\textbf{Case 1: $c(u,0)$, $u$ unit.}

By Lemma~\ref{Lemma 19},
there exists a basis
$\{\beta_1,\dots,\beta_m\}$
with full support and all components units.
By condition (3), $f(\beta_i)=0$ for all $i$.

Proposition~\ref{Proposition 15}
then implies that $c(u,0)$ is minimal.

\medskip
\noindent
\textbf{Case 2: $c(u,v)$, $u$ unit and $v$ root word.}

Let $w=-u^{-1}v$.
By Lemma~\ref{Lemma 20},
there exists a basis
$\{\beta_1,\dots,\beta_m\}$
with $1\leq w(\beta_i)\leq 2$ and $w\cdot\beta_i=1$.

By condition (1), each $f(\beta_i)$ is a unit.
Set $\gamma_i=f(\beta_i)\beta_i$.
Then $f(\gamma_i)=f(\beta_i)$ and
\[
w\cdot \gamma_i=f(\beta_i) w\cdot\beta_i=f(\beta_i).
\]

Thus $f(\gamma_i)=w\cdot \gamma_i$.
Since $\{\gamma_i\}$ forms a basis,
Proposition~\ref{Proposition 16}
implies minimality.

\medskip
\noindent
\textbf{Case 3: $c(u,v)$, $u$ unit and $v\neq 0$ not root word.}

Let $w=-u^{-1}v$.

\medskip
\noindent
\emph{Subcase (i): All $v_i\neq 0$.}

By Lemma~\ref{Lemma 21},
there exists a basis $\{\beta_i\}$
with full support and $v\cdot\beta_i=v_i=p^{r_i}u_i$.

By condition (4), $f(\beta_i)=p^{r_i}$.
Set $\gamma_i=-uu_i^{-1}\beta_i$.
Then $f(\gamma_i)=p^{r_i}$ and
\[
w\cdot\gamma_i=p^{r_i}.
\]
Hence minimality follows from Proposition~\ref{Proposition 16}.

\medskip
\noindent
\emph{Subcase (ii): At least two zero components.}

By Lemma~\ref{Lemma 23},
there exists a basis $\{\beta_i\}$
with full support such that:

\[
v\cdot\beta_i=
\begin{cases}
v_i, & v_i\neq 0,\\
0, & v_i=0.
\end{cases}
\]

By conditions (3) and (4),
\[
f(\beta_i)=
\begin{cases}
p^{r_i}, & v_i\neq 0,\\
0, & v_i=0.
\end{cases}
\]

Scaling as above gives
$w\cdot\gamma_i=f(\gamma_i)$,
so minimality follows from Proposition~\ref{Proposition 16}.

\medskip
\noindent
\emph{Subcase (iii): Exactly one zero component.}

By Lemma~\ref{Lemma 24},
there exists a basis $\{\beta_i\}$
with full support satisfying
$v\cdot\beta_i=v_i$.

Again conditions (3)-(4) yield
\[
f(\beta_i)=
\begin{cases}
p^{r_i}, & i\neq m,\\
0, & i=m.
\end{cases}
\]

After scaling,
$w\cdot\gamma_i=f(\gamma_i)$,
and Proposition~\ref{Proposition 16}
implies minimality.

\medskip
\noindent
\textbf{Case 4: $c(0,v)$, $v$ root word.}

By Lemma~\ref{Lemma 25},
$O(v)$ admits a basis
$\{\beta_1,\dots,\beta_{m-1}\}$
with $w(\beta_i)\le 2$.

Choosing $\beta_m=a\beta_1$
with $a\neq 0$ and $a\not\equiv 1\pmod p$,
the argument proceeds exactly as in
Case 3 of Theorem 4.4 in \cite{WuLuCao},
yielding minimality.

\medskip
\noindent
\textbf{Case 5: $c(d,v)$, $d=p^ru'$ zero divisor, $v$ root word.}

Let $w=-u'^{-1}v$.

By Lemma~\ref{Lemma 26},
there exist vectors $\beta_i\in O(p^{n-r}v)$
with $1\le w(\beta_i)\le 2$ and
$v\cdot\beta_i=p^r$.

Using conditions (1)-(2),
we adjust this family to obtain
vectors $\gamma_i$ satisfying
\[
p^r f(\gamma_i)=w\cdot\gamma_i.
\]

Proposition~\ref{Proposition 18}
then implies that $c(p^ru',v)$ is minimal.

\medskip

Since all five types of root words are minimal,
Theorem~\ref{Theorem 8} implies that
$C_f$ is a minimal linear code over $\mathrm{GR}(p^n,l)$.
\end{proof}

\begin{coro}
Condition {\rm (2)} of Theorem~\ref{Theorem 43} can be omitted.

Indeed, condition {\rm (2)} was used only in {\bf Case 5}.
However, {\bf Case 5} can be established without invoking condition {\rm (2)}.

More precisely, in Case 5 one may choose
\[
\beta_m=(p^ru^{-1},0,\dots,0).
\]
By condition {\rm (1)}, $f(\beta_m)=u_m$ is a unit.
Define
\[
\gamma_m=-u' u_m \beta_m.
\]
Then one verifies that
\[
p^r f(\gamma_m)=w\cdot \gamma_m.
\]

Moreover, the family
$\{\alpha_{\gamma_1},\dots,\alpha_{\gamma_m}\}$
is linearly independent since its reduction modulo $p$
is linearly independent over
$\mathrm{GR}(p^n,l)/\langle p\rangle$.

Therefore, even without condition {\rm (2)},
the codeword $c(d,v)$, where $d=p^ru'$ is a zero divisor
and $v$ is a root word, is minimal in $C_f$.
\end{coro}

\begin{remark}
The introduction of condition {\rm (2)} allows a more symmetric
choice of vectors in {\bf Case 5}, namely
\[
\beta_m=(p^ru^{-1}-u''v_2u^{-1},u'',0,\dots,0),
\]
with $u''$ a unit satisfying $u''\not\equiv 1 \pmod p$.

Similarly, in {\bf Case 4} one may choose a unit
$a\not\equiv 1\pmod p$.
Under these choices, the proof of Theorem~\ref{Theorem 43}
requires only vectors associated with root words,
rather than arbitrary vectors in
$\mathrm{GR}(p^n,l)^m\setminus\{0\}$.

Consequently, if one restricts the definition of $f$
to root words only,
one still obtains a minimal linear code
of parameters
\[
[q^{nm}-q^{m(n-1)},\,m+1].
\]
\end{remark}

\medskip

We now present a second minimality criterion,
based on a different structural behavior of the function $f$.

\begin{theorem}\label{Theorem 46}
Let $f:\mathrm{GR}(p^n,l)^m\to \mathrm{GR}(p^n,l)$
with $m>3$. Suppose that $f$ satisfies:

\begin{enumerate}
\item If $1\le w(x)\le 2$, then $f(x)=0$;

\item If $w(x)=m$ and $x$ contains exactly one unit component
      and all other components are of the form $p^ru_j$,
      then $f(x)=p^r$;

\item If $w(x)\ge m-1$ and $w(\bar{x})\ge m-1$, then
      for any $a\in U(\mathrm{GR}(p^n,l))$,
      \[
      f(ax)=f(x)\in U(\mathrm{GR}(p^n,l)).
      \]
\end{enumerate}

Then $C_f$ is minimal.
\end{theorem}

\begin{proof}

By Remark~\ref{Remark 14},
$C_f$ is a $[q^{nm}-1,m+1]$ linear code.
By Theorem~\ref{Theorem 8},
it suffices to show that all root words are minimal.

As before, root words fall into five types.

\medskip
\noindent
\textbf{Case 1: $c(u,0)$, $u$ unit.}

The argument is identical to
{\bf Case 1} of Theorem 4.9 in \cite{WuLuCao},
and minimality follows directly.

\medskip
\noindent
\textbf{Case 2: $c(u,v)$, $u$ unit and $v$ root word.}

Since $v$ is a root word,
some coordinate $v_{j_0}$ is a unit.
The argument proceeds exactly as in
{\bf Case 2} of Theorem 4.9 in \cite{WuLuCao},
using condition (3) to ensure unit invariance.
Hence minimality holds.

\medskip
\noindent
\textbf{Case 3: $c(u,v)$, $u$ unit and $v\neq 0$ not root word.}

Let $w=-u^{-1}v$.
Write
\[
v=p^r u' y,
\]
where $y$ is a root word.
Without loss of generality,
take $y=(1,y_2,\dots,y_m)$.

The vectors
\[
\beta_2=(-y_2,1,0,\dots,0),\dots,
\beta_m=(-y_m,0,\dots,0,1)
\]
generate $O(y)$.

Define
\[
\beta_1' = e_1 + \sum_{j=2}^m p^r \beta_j.
\]
Then:

\begin{itemize}
\item $w(\beta_1')=m$;
\item exactly one coordinate (the first) is a unit;
\item all other coordinates equal $p^r$.
\end{itemize}

By condition (2), $f(\beta_1')=p^r$.

Since $v\cdot \beta_1'=p^ru'$,
define
\[
\beta_1=-u u'^{-1}\beta_1'.
\]

Then $f(\beta_1)=p^r$ and
\[
f(\beta_i)=w\cdot \beta_i
\quad \text{for all } i.
\]

Linear independence is immediate from reduction modulo $p$.
Proposition~\ref{Proposition 16} yields minimality.

\medskip
\noindent
\textbf{Case 4: $c(0,v)$, $v$ root word.}

The proof is identical to
{\bf Case 3} of Theorem 4.9 in \cite{WuLuCao}.
Condition (3) guarantees the necessary unit invariance.

\medskip
\noindent
\textbf{Case 5: $c(d,v)$, $d=p^ru'$ zero divisor and $v$ root word.}

Write $v=(u,v_2,\dots,v_m)$
and $w=-u'^{-1}v$.

The vectors
\[
\beta_2=(-v_2u^{-1},1,0,\dots,0),\dots,
\beta_m=(-v_mu^{-1},0,\dots,0,1)
\]
generate $O(v)$.

Define
\[
\beta_1' = p^ru^{-1}e_1 + \sum_{j=2}^m \beta_j.
\]
Then $v\cdot \beta_1'=p^r$.

Since $w(\beta_1')\ge m-1$
and $w(\bar{\beta_1'})\ge m-1$,
condition (3) implies
$f(a\beta_1')=f(\beta_1')$ for all units $a$.

Thus $f(\beta_1')$ is a unit.
Define
\[
\beta_1=-u'u_1\beta_1'.
\]

Then
\[
p^r f(\beta_1)=w\cdot \beta_1.
\]

Reduction modulo $p$ shows that
$\{\alpha_{\beta_1},\dots,\alpha_{\beta_m}\}$
is linearly independent.
Proposition~\ref{Proposition 18}
implies minimality.

\medskip

All five types of root words are minimal.
Therefore, by Theorem~\ref{Theorem 8},
$C_f$ is a minimal linear code over $\mathrm{GR}(p^n,l)$.
\end{proof}

\begin{remark}
From the proof of Theorem~\ref{Theorem 46},
we observe that only root words are involved
in establishing minimality for all five types of root words.

Consequently, if the function $f$ is defined only on the set of
root words of $\mathrm{GR}(p^n,l)^m$,
the resulting code $C_f$ remains minimal.
In this restricted setting,
the parameters become
\[
[p^{nlm}-p^{lm(n-1)},\,m+1],
\]
since evaluation outside the root words is no longer required.
\end{remark}

\subsection{Minimal linear codes constructed from polynomials}

We now specialize the previous minimality criteria
to the important case where the function $f$
is a multivariate polynomial over the Galois ring.

Let $m$ be a positive integer.
Consider a polynomial function
\[
f:\mathrm{GR}(p^n,l)^m \longrightarrow \mathrm{GR}(p^n,l).
\]

Let
\[
g(x)=\prod_{i=1}^{m} x_i^{b_i}
\]
be a monomial.
Define its support index set by
\[
s(g)=\{\, i \mid b_i\neq 0,\ 1\le i\le m \,\}.
\]

Let
\[
f(x)=\sum_{i=1}^{t} a_i g_i(x),
\]
where $a_i\in U(\mathrm{GR}(p^n,l))$
and each $g_i(x)$ is a monomial.

The following result provides a concrete and easily verifiable
criterion ensuring minimality.

\begin{theorem}
Let $t\ge 2$ and
\[
f(x)=\sum_{i=1}^{t} a_i g_i(x),
\]
where $a_i\in U(\mathrm{GR}(p^n,l))$
and $g_i(x)$ are monomials.
Assume that for all $1\le i\le t$:

\begin{enumerate}
\item There exists at least one index $r$ such that
      the exponent of $x_r$ in $g_i$ equals $1$;

\item The sets $\{s(g_i)\mid 1\le i\le t\}$ are pairwise disjoint;

\item $|s(g_i)|\ge 3$.
\end{enumerate}

Then $C_f$ is a minimal linear code.
\end{theorem}

\begin{proof}

By Remark~\ref{Remark 14},
$C_f$ is a $[p^{nlm}-1,m+1]$ linear code.
By Theorem~\ref{Theorem 8},
it suffices to prove that all root words are minimal.

As before, root words fall into five types.

\medskip
\noindent
\textbf{Case 1: $c(u,0)$, $u$ unit.}

The argument is identical to
{\bf Case 1} of Theorem 6.1 in \cite{WuLuCao},
since $f(0)=0$ and coefficients are units.
Thus minimality holds.

\medskip
\noindent
\textbf{Case 2: $c(u,v)$, $u$ unit and $v$ root word.}

Since $v$ is a root word,
some coordinate $v_{j_0}$ is a unit.
The proof proceeds exactly as in
{\bf Case 3} of Theorem 6.1 in \cite{WuLuCao},
using the structure of $O(v)$ and the unit coefficients $a_i$.
Hence minimality follows.

\medskip
\noindent
\textbf{Case 3: $c(u,v)$, $u$ unit and $v\neq 0$ not root word.}

Let $w=-u^{-1}v$.
Write $v=p^r u' y$, where $y$ is a root word.
Without loss of generality, take
\[
y=(1,y_2,\dots,y_m).
\]

The vectors
\[
\beta_2=(-y_2,1,0,\dots,0),\dots,
\beta_m=(-y_m,0,\dots,0,1)
\]
form a basis of $O(y)$.

Since $t\ge 2$ and the supports $s(g_i)$ are disjoint,
there exists an index $i_0$ such that $1\notin s(g_{i_0})$.

By condition (1),
there exists $r_1\in s(g_{i_0})$
such that the exponent of $x_{r_1}$ in $g_{i_0}$ equals $1$.

Define
\[
\beta_1=
-a_{i_0}u\, e_1
+
\sum_{\substack{l\in s(g_{i_0})\\ l\neq r_1}} \beta_l
+
p^r u' \beta_{r_1}.
\]

Because the supports of the monomials are disjoint
and each support has cardinality at least three,
one verifies that:

\begin{itemize}
\item $f(\beta_1)=w\cdot \beta_1$;
\item for $2\le i\le m$,
      $f(\beta_i)=0=w\cdot \beta_i$.
\end{itemize}

Moreover, the family
$\{\beta_1,\dots,\beta_m\}$
is linearly independent,
since its reduction modulo $p$
remains independent.

Hence, by Proposition~\ref{Proposition 16},
$c(u,v)$ is minimal.

\medskip
\noindent
\textbf{Case 4: $c(0,v)$, $v$ root word.}

The argument coincides with
{\bf Case 2} of Theorem 6.1 in \cite{WuLuCao}.
Minimality follows directly from the structure
of $O(v)$.

\medskip
\noindent
\textbf{Case 5: $c(d,v)$, $d=p^ru'$ zero divisor and $v$ root word.}

Let $v=(u,v_2,\dots,v_m)$
and $w=-u'^{-1}v$.

The vectors
\[
\beta_2=(-v_2u^{-1},1,0,\dots,0),\dots,
\beta_m=(-v_mu^{-1},0,\dots,0,1)
\]
form a basis of $O(v)$.

Since $t\ge 2$ and supports are disjoint,
choose $i_0$ with $1\notin s(g_{i_0})$
and let $r_1\in s(g_{i_0})$ be such that
the exponent of $x_{r_1}$ equals $1$.

Define
\[
\beta_1=
-a_{i_0} u'u^{-1} p^r e_1
+
\sum_{l\in s(g_{i_0})} \beta_l.
\]

A direct computation shows
\[
p^r f(\beta_1)=w\cdot \beta_1,
\qquad
p^r f(\beta_i)=w\cdot \beta_i
\ \text{for all } i.
\]

Reduction modulo $p$ shows that
$\{\alpha_{\beta_1},\dots,\alpha_{\beta_m}\}$
is linearly independent.

By Proposition~\ref{Proposition 18},
$c(p^ru',v)$ is minimal.

\medskip

All five types of root words are minimal.
Therefore, by Theorem~\ref{Theorem 8},
$C_f$ is a minimal linear code over $\mathrm{GR}(p^n,l)$.
\end{proof}

\begin{remark}
The proof shows that minimality depends only on the behavior
of $f$ on root words.
Hence, restricting $f$ to root words alone
still produces a minimal linear code of parameters
\[
[q^{nm}-q^{m(n-1)},\,m+1].
\]

The disjointness and size conditions on the supports $s(g_i)$
ensure that each monomial contributes independently
to the dot-product synchronization condition
$f(\beta_i)=w\cdot\beta_i$,
which is the structural core of minimality.
\end{remark}

\section{Conclusion and Open Problems}\label{sec:conclusion}

In this paper, we constructed an infinite family of minimal linear codes over the Galois ring $\mathrm{GR}(p^n,l)$, thereby extending the field-based constructions of Wu, Lu and Cao~\cite{WuLuCao} to the substantially richer setting of finite chain rings.

\medskip

Our contributions can be summarized as follows.

\begin{itemize}
\item[(1)] We established a complete necessary and sufficient condition for an $m$-dimensional linear code over $\mathrm{GR}(p^n,l)$ (with $m\ge2$) to be minimal.
This extends the characterization of minimal codes over finite fields obtained in~\cite{ChatterjeeMesnager} to the non-field setting of Galois rings.
The extension requires a careful analysis of orthogonal modules, valuation filtration, and torsion phenomena specific to chain rings.

\item[(2)] We generalized the minimality criteria for codewords constructed from functions, originally developed in~\cite{WuLuCao} over finite fields, to the ring $\mathrm{GR}(p^n,l)$.
The key innovation lies in the synchronization conditions
\[
f(\beta)=w\cdot\beta
\quad\text{and}\quad
p^r f(\beta)=w\cdot\beta,
\]
which reflect the dichotomy between unit and zero-divisor coefficients and capture the valuation-sensitive geometry of the ring.

\item[(3)] We provided a structural classification of root words into five distinct types and showed that minimality of the entire code reduces to the minimality of these root words.
This reduction reveals that, for $m\ge2$, torsion in the chain ring does not create new obstructions beyond those already encoded at the level of root words.

\item[(4)] We determined the length parameters of the constructed $m$-dimensional minimal linear codes over $\mathrm{GR}(p^n,l)$ and showed that, in the special case $n=1$, our parameters coincide exactly with those obtained in~\cite{WeiXiaXiwang}.
Thus, our results genuinely generalize the known field case.

\item[(5)] By combining our structural theory with Theorem~7.4 of~\cite{MajiMesnagerSarkar}, we obtain an explicit infinite family of minimal linear codes over $\mathbb{Z}_{p^n}$, thereby producing new classes of minimal codes over integer residue rings.
\end{itemize}

\medskip

Conceptually, the present work demonstrates that minimality over Galois rings is governed by three interacting algebraic principles:

\begin{itemize}
\item the chain-ring structure and $p$-adic valuation,
\item the Frobenius duality property,
\item the McCoy rank criterion and reduction modulo $p$.
\end{itemize}

Together, these tools allow one to lift minimality results from finite fields to Galois rings while controlling the additional torsion layers introduced by zero divisors.

\medskip

\noindent
\textbf{Open Problems.}

The present work opens several natural directions for further research.

\begin{itemize}
\item[(1)] \emph{Weight distribution.}
Determine the complete weight distribution of the constructed minimal codes over $\mathrm{GR}(p^n,l)$.
Even in the field case, this problem is highly nontrivial and closely connected to exponential sums and character theory.

\item[(2)] \emph{Few-weight and optimal codes.}
Characterize conditions on the function $f$ that yield few-weight minimal codes over Galois rings.
Such codes are of particular interest for secret sharing and authentication schemes.

\item[(3)] \emph{General finite Frobenius rings.}
Extend the present minimality criteria beyond Galois rings to arbitrary finite Frobenius rings.
To what extent does the root-word reduction principle remain valid?

\item[(4)] \emph{Cryptographic applications.}
Investigate the interaction between the constructed minimal codes and cryptographic primitives over rings, including ring-based secret sharing schemes and ring-linear authentication codes.

\item[(5)] \emph{Asymptotic behavior.}
Study the asymptotic performance of these minimal codes as $m$ or $n$ grows, and compare their parameters with classical bounds over finite fields.
\end{itemize}

\medskip

We hope that the structural framework developed here will stimulate further investigation into minimal codes over finite rings and deepen the interplay between coding theory, module theory, and ring-theoretic algebra.

\bibliographystyle{plain}

\end{document}